\documentclass[a4paper]{article}

\usepackage{a4wide}
\usepackage{latexsym}
\usepackage[UKenglish]{babel}
\usepackage{graphicx}
\usepackage{algorithm}
\usepackage{algpseudocode}
\graphicspath{{./}{figures/}}
\usepackage[colorlinks=true,linkcolor=blue,citecolor=blue]{hyperref}
\usepackage{xcolor}
\usepackage{amsfonts,amsmath,amsthm,mathtools,bbm,cool}
\usepackage{dsfont}
\usepackage{verbatim}
\usepackage{tikz}
\usepackage{subcaption}
\usepackage{graphicx} 
\usepackage{array}
\usepackage{booktabs}
\usepackage{booktabs}

\usepackage{tabularx}
\newtheorem{thm}{Theorem}[section]

\newtheorem{prop}[thm]{Proposition}
\theoremstyle{definition}

\theoremstyle{remark}

\theoremstyle{remark}

\newcommand{\be}{\begin{equation}}
\newcommand{\ee}{\end{equation}}

\newcommand{\eps}{\varepsilon}

\renewcommand{\P}{\mathcal{P}}
\newcommand{\p}{\partial}
\newcommand{\M}{\mathcal{M}}

\allowdisplaybreaks

\begin{document}

\title{
Kinetic SIS opinion-driven models with asymmetric awareness feedback: macroscopic limit and polarization}

\author{
		J. P. Pinasco\\
		{\small	Department of Mathematics} \\
		{\small Universidad de Buenos Aires, Argentina} \\
		{\small\tt jpinasco@dm.uba.ar} 
		\\
		Nicolas Saintier\\
		{\small	Department of Mathematics} \\
		{\small Universidad de Buenos Aires, Argentina} \\
		{\small\tt nsaintie@dm.uba.ar} 
		\\
		Horacio Tettamanti\\
		{\small	Department of Mathematics ``F. Casorati''} \\
		{\small University of Pavia, Italy} \\
		{\small\tt horacio.tettamanti01@universitadipavia.it }
		\\
		 Mattia Zanella \\
		{\small	Department of Mathematics ``F. Casorati''} \\
		{\small University of Pavia, Italy} \\
		{\small\tt mattia.zanella@unipv.it} 
		}
		
\date{}
\maketitle

\begin{abstract}
We study a kinetic multi-agent framework coupling opinion dynamics with epidemic spreading, where individual social behaviour both affects and is affected by disease transmission. Each agent is characterised by an epidemiological state and a continuous opinion variable measuring compliance with non-pharmaceutical interventions. The key mechanism of the model is an asymmetric opinion update driven by epidemic encounters: infection events induce more cautious attitudes, while failed transmissions push individuals toward more extreme opinions. We focus on a prototypical SIS setting, for which we derive a macroscopic kinetic description and, in a fast social-interaction regime, a reduced system of differential equations capturing the feedback between epidemic prevalence and opinion evolution. Convergence of the reduced model is rigorously quantified through a modified Wasserstein distance. Numerical simulations highlight how infection-induced awareness and non-infection-driven extremization jointly shape collective epidemic-opinion dynamics.
\end{abstract}

\section{Introduction} \label{sect:intro}

In recent years there has been an extensive interest in capturing the interaction between communities' opinions and the threat of the epidemic, see e.g. \cite{Glanz} and the references therein. Indeed, the spread of infectious diseases is not driven solely by biological and environmental mechanisms, but is deeply shaped by social responses emerging from individual interactions \cite{BDM,BDMOG}. In particular, individual behaviour and risk perception evolve through microscopic opinion formation processes, such as social influence, information exchange, and personal experience, which may lead to the emergence of collective consensus or polarization around non-pharmaceutical interventions (NPIs) \cite{Block,FGWJ,LG,GGA,PLBNB}. Since emergent social patterns have an impact on epidemic trajectories, capturing the feedback between disease dynamics and the underlying mechanisms of opinion formation has become a central challenge in contemporary epidemiological modelling. 

Classical compartmental models describe disease propagation by partitioning the population into epidemiological classes whose temporal evolution is governed by transition rates \cite{Die,IP}. While compartmental models have proven effective in capturing macroscopic epidemic trends, their formulation ultimately are not inferred from microscopic processes driven by elementary social forces. Several works concentrate on their derivation to have access to a multiscale view on epidemic dynamics, see \cite{BP,CMRV,CGMR,CPS,DPTZ,GPS,KGPS,MTZ,Z2}. In the context of opinion dynamics, individual attitudes toward preventive measures vary continuously and evolve through social interactions, information exchange, and personal experience with the disease. In particular, compliance with NPIs cannot be adequately represented by discrete compartments alone. Unlike vaccination status, which may naturally correspond to a distinct epidemiological class, adherence to protective measures is better described as a continuous behavioural trait, shaped by social influence and affected by the epidemic state itself. This observation calls for modelling approaches capable of resolving heterogeneity at the individual level while remaining consistent with population-scale epidemic descriptions  \cite{BL,DMLM,LT2,NWZ}.

Motivated by recent advances in kinetic models for consensus formation and opinion dynamics developed in \cite{ACDZ,BW,BZ,BTZ,DFWZ,Z}, we introduce a multi-agent framework that couples opinion evolution with disease transmission with direct feedback from the epidemic dynamics. Individuals are characterised simultaneously by their epidemiological status and by a continuous opinion variable representing their attitude toward preventive measures. Social interactions drive opinion changes through pairwise exchanges, while epidemic processes depend on these opinions and, in turn, influence them. In particular, the contagion rate depends on individuals’ attitudes toward protection, while infection and recovery events modify their propensity to comply with NPIs. The resulting microscopic dynamics are formulated in a probabilistic setting and coupled with a standard SIS epidemic structure. By performing a kinetic upscaling, we derive a macroscopic description in terms of opinion distributions within the susceptible and infected populations. 

To gain analytical insight, we investigate a regime characterised by a separation of time scales, in which social interactions occur more frequently than epidemic transitions. In this setting, we show that the kinetic system can be approximated by a reduced system of ordinary differential equations describing the evolution of macroscopic quantities, such as the number of infected individuals and the mean opinion in each compartment. The convergence toward this reduced description is established via a modified Wasserstein distance, yielding explicit convergence rates under general assumptions on the interaction mechanisms. The reduced system enables a rigorous analysis of the coupled opinion-disease dynamics.  We stress that, although the SIS model is adopted as a guiding example, the proposed framework naturally extends to more comple epidemiological structures, such as SIR or SEIR models, providing a flexible setting to investigate the role of social behaviour in epidemic dynamics and laying the groundwork for future studies on the optimal control of opinion to reduce disease transmission.

The paper is organised as follows: in Section \ref{sect:kinetic.intro} we introduce an epidemiological modelling approach which takes into account microscopic opinion interaction dynamics, in Section \ref{sect:no_selfthinking} we analise the resulting model in the absence of self-thinking providing well-posedness of the model  together with a first order macroscopic model for mass exchanges. Furthermore, we provide an argument for rigorous convergence toward a reduced complexity macroscopic model encapsulating information on the evolution of the first two moments, a formal analysis is considered in the case of self-thinking. Finally, in Section \ref{sect:num} we provide several numerical tests to understand the impact of an awareness feedback on epidemiologically observable quantities.

\section{Kinetic approach to compartmental epidemiology}\label{sect:kinetic.intro}

In this section we introduce a Susceptible-Infected (SIS) model coupled with the opinion of individuals considering various types of interactions that may affect the epidemic state as well as the opinion of each agent. In this model, each individual may alternatively belong to the susceptible compartment $(S)$, corresponding to the agents that may contract the disease, or the $(I)$ compartment corresponding to infectious and infected agents. 
Furthermore, each individual is characterised by a continuous opinion variable $w \in I$, with $I = [-1,1]$, where the extreme values $\pm 1$ denote two opposite beliefs regarding protective behaviour. By convention, $w = -1$ corresponds to an opinion opposed to protective behaviour, while $w = +1$ represents a fully favourable opinion toward protective behaviour. We assume that the probability of contagion is affected by the opinion of each individual.

In order to describe the coupled dynamics, we introduce $f^{S}_t,f^I_t:\mathbb{R}_+ \rightarrow \mathbb{R}_+$ the distribution of opinions at a certain time $t \geq 0$ in the epidemic compartment $S$ and $I$. 
Therefore $f^{H}_t(w)dw$ represents the fraction of agents at time $t\ge0$ belonging to the compartment $H\in \mathcal{C}:=\{S,I\}$ and having an opinion in $[w,w+dw]$. 
Notice that $f^{S}_t$ and $f^I_t$ are in general non-negative measures on $[-1,1]$ that we write as functions for notational simplicity only. 
We also consider $f_t$ the distribution of opinion in the whole population. It is the probability measure on $[-1,1]$ defined as 
\begin{equation}
    f_t(w) := \sum_{H \in \mathcal C} f^H_t(w),  \hspace{0.5cm} \int_{-1}^{1}  f_t(w)dw = 1. 
\end{equation}
The fraction $\rho^H_t$ of individual in the compartment $H$ and the r-th moment of their opinion opinion $m^{H,r}_t$ are then
\begin{equation}
    \rho^H_t=\int_{-1}^{1} f^H_t(w)dw, \hspace{0.5cm} \rho^H_tm^{H,r}_t=\int_{-1}^{1} w^r f^H_t(w)dw, \qquad H \in \mathcal C = \{S,I\}. 
\end{equation}
For simplicity of notation, we will denote the mean opinion in compartment $H\in \mathcal C$ as $m^H_t$. 

\subsection{Microscopic interaction dynamics}

The interaction mechanism of the multi-agent system under study is characterised by both physical interactions, taking place through a stylised contact dynamics, and social interactions, taking place on social media. While changes in the epidemiological compartments can only occur through physical encounters between agents, opinion changes influence their behaviour.  

Following the work presented by \cite{BZ,Z}, we consider that the probability of contagion $\kappa(w,w_*)$ between two agents with opinion $w$ and $w_*$ is given by
\begin{equation}\label{ContagionRate}
 \kappa(w,w_*) = \frac{\beta}{4^\alpha}(1-w)^\alpha(1-w_*)^\alpha,
\end{equation}
where $\alpha \geq 0$ measures the level of coupling between both opinions and the epidemic dynamics, 
and $\beta \in [0,1]$ is the contagion probability of the epidemics in absence of preventive measures. 
We can easily notice that this quantity is maximized when the interacting agents totally reject the protective measures, and is minimized 
when they fully adopt them. In the limiting case where $\alpha=0$, the transmission dynamics is independent of the opinions. 

Therefore, whether a \textbf{physical interaction} between an agent in compartment $I$ and an agent in compartment $S$ occurs depends both on the biological characteristics of the epidemic and on the individuals' opinions toward protective measures. 
Moreover, the S agent will modify his opinion depending on the results of the interaction. We consider two different possible cases: 
\begin{enumerate}
\item[$i$)] The agent in $S$ may not get infected, which occurs with probability $1-\kappa(w,w_*)$. 
He then adopts a more reckless attitude modifying  his original opinion $w$ to a new opinion $w'_{S\rightarrow S}$ given by 
\begin{equation}
\label{eq:micro_SS}
w'_{S\rightarrow S} = w + h(-1-w).  
\end{equation}
Thus $w'_{S\rightarrow S}$ is obtained by shifting $w$ toward $-1$ at a pace controlled by the parameter $h\in [0,1]$. 
In this case only the opinion inside the S compartment is modified. 

\item[$ii$)] On the opposite, the agent in $S$ may get infected, which occurs with probability $\kappa(w,w_*)$ defined in \eqref{ContagionRate}. 
In that case, he switches compartment becoming $I$ while adopting a new opinion $w'_{S\rightarrow I}$ given by
\begin{equation}
\label{eq:micro_SI}
w'_{S\rightarrow I} = w + h(1-w).  
\end{equation}
Thus, upon getting infected, the agent   adopts a more compliant attitude toward preventive measures by moving his former opinion toward 1. 
Notice that if $h=0$, then $w'_{S\rightarrow I} = w$ i.e. $S$ becomes $I$ without changing his opinion, which is the case considered in \cite{BZ}.

\end{enumerate}

Apart from direct interactions among individuals, we introduce two additional mechanisms that affect agents' opinions and their attitude toward protective measures. These mechanisms account for individual-level behavioural changes that are not directly triggered by pairwise interactions but arise from time-dependent effects and epidemiological transitions.

\begin{itemize}
\item \textbf{Memory fading}. 
Individuals may gradually become more careless over time, leading to a progressive reduction in their adherence to protective measures. This effect models the loss of risk perception due to habituation or fatigue and is described by the update rule
\begin{equation}\label{micro:memory}
w'_H := w +  h(-1-w), \qquad H \in \mathcal C,
\end{equation}
where $ h \in [0,1]$ quantifies the intensity of the memory fading mechanism and $\mathcal C=\{S,I\}$ denotes the set of epidemiological compartments.

\item \textbf{Recovery}. 
We assume that a change of opinion may also occur when an infected individual recovers and its epidemiological state transitions from $I$ to $S$. Upon recovery, individuals may perceive themselves as less vulnerable and thus adopt a more careless attitude toward protective measures. This effect is modelled by
\begin{equation}\label{micro:recovery}
w'_{I\rightarrow S} = w +  h(-1-w),
\end{equation}
where $ h \in [0,1]$ represents the strength of the opinion shift induced by recovery.
\end{itemize}

The parameters used to control the strength of the opinion change in \eqref{eq:micro_SS}, \eqref{eq:micro_SI}, \eqref{micro:memory}, \eqref{micro:recovery} 
could depend both on the mechanism and the original opinion $w$. We chose to use the same constant parameter $h$ for simplicity. 

Finally, a pure opinion-type dynamics is considered to define the information exchange between agents. Social media interactions constitute a paradigmatic example of purely social dynamics. While the epidemiological state of an individual may influence its behaviour, interactions occurring in the virtual space do not give rise to direct contagion events. Hence, such interactions solely contribute to the evolution of agents' opinions. Building upon classical kinetic models for opinion formation \cite{T,Z2}, we consider a binary interaction between two agents characterised by opinions $(w,w_*)\in [-1,1]\times[-1,1]$. 
The post-interaction opinions $(w',w_*')$ are determined by the superposition of an aggregation process that tends to reduce opinion disparities between interacting agents, and a stochastic self-thinking mechanism 
modeling random opinion changes that may be caused by factors not contemplated in our model such as the internal process of information within each individual:

\begin{equation}\label{Def_wCC}
\begin{aligned}
 w' &= w + \xi  P(w,w_*)(w_*-w) + D(w,w_*)\eta, \\
 w' &= w + \xi  P(w_*,w)(w-w_*) + D(w_*,w)\eta_*.
 \end{aligned}
 \end{equation}
The nonnegative function $P$ measures the compromise forces and is tuned by $\xi>0$. 
The iid random variables  $\eta$ and $\eta_*$ have mean 0 and variance $\sigma^2$. 
They model the self-thinking process 
 whose strength is controlled by the nonnegative function $D$.

\subsection{The kinetic model}

By considering different types of mechanisms and interactions that affect both the opinion and epidemiological state of the agents, 
$f_t^S$ and $f_t^I$ evolve according to the following coupled system of kinetic-type equations: 

\begin{equation}\label{full.system}
\begin{aligned}
\dfrac{\partial}{\partial t}f_t^S &= Q_{S\to S}(f_t^S) + Q^G_{I\to S}(f_t^I)  - Q^L_{S+I\rightarrow I+I}(f_t^S,f_t^I)  
+ Q_{S+I\rightarrow S+I}(f_t^S,f_t^I) + \sum_{H\in \mathcal C}Q_{SH}(f_t^S,f_t^H)\\
\dfrac{\partial}{\partial t} f_t^I &= Q_{I\to I}(f_t^I) - Q^L_{I\to S}(f_t^I)  + Q^G_{S+I\rightarrow I+I}(f_t^S,f_t^I)  
+  \sum_{H\in \mathcal C}Q_{IH}(f_t^I,f_t^H), 
\end{aligned}
\end{equation}
Let $\phi(w)$ denote a test function. The operator $Q_{H\to H}(f^H_t) = Q_{H\to H}(f^H_t)(w,t)$ takes into account memory fading dynamics
\[
\int_{{[-1,1]}}\phi(w)Q_{H\to H}(f^H_t)dw = \int_{{[-1,1]}}(\phi(w_H^\prime)-\phi(w)) df^H_t(w), \qquad H \in \mathcal C, 
\]
being $w_H^\prime$ defined in \eqref{micro:memory}. 
The recovery process is encapsulated in the operators $Q^G_{I\to S}(f^I_t) = Q^G_{I\to S}(f^I_t)(w,t)$, $Q^L_{I\to S}(f^I_t) = Q^L_{I\to S}(f^I_t)(w,t)$, which represent the gain for the $S$ compartment and loss in the $I$ compartment, respectively. These operators are defined as follows
\[
\begin{split}
\int_{[-1,1]} \phi(w) Q^G_{I\to S}(f_t^I)(w)dw &= \int_{[-1,1]} \phi(w^\prime_{I\to S})df_t^I(w) \\
\int_{[-1,1]} \phi(w) Q^L_{I\to S}(f_t^I)(w)dw &= \int_{[-1,1]} \phi(w)df_t^I(w),
\end{split}\]
where $w^\prime_{I\to S}$ has been defined in \eqref{micro:recovery}. 
The operator $Q_{S+I\to S+I}(f_t^S,f_t^I) = Q_{S+I\to S+I}(f_t^S,f_t^I)(w)$ describes interactions in which agents in $S$ do not change compartment after interaction with agents in $I$ and has therefore the following form 
\[
\int_{[-1,1]} \phi(w) Q_{S+I\to S+I}(f_t^S,f_t^I)(w)dw = \int_{[-1,1]\times [-1,1]} (\phi(w^\prime_{S\to S}) - \varphi(w))(1-\kappa(w,w_*))df_t^S(w)df_t^I(w_*),
\]
where $w^\prime_{S\to S}$ has been defined in \eqref{eq:micro_SS} and $\kappa(\cdot,\cdot)$ in \eqref{ContagionRate}. 
Similarly, gain resulting from agents in $S$ that get infected are encoded into $Q^G_{S+I\to I+I}(f_t^S,f_t^I) = Q^G_{S+I\to S+I}(f_t^S,f_t^I)(w)$
\[
\int_{[-1,1]} \phi(w) Q^G_{S+I\to I+I}(f_t^S,f_t^I)(w)dw = \int_{[-1,1]} \phi(w^\prime_{S\to I}) \kappa(w,w_*)df_t^S(w)df_t^I(w_*),
\]
where $w^\prime_{S\to I}$ has been defined in \eqref{eq:micro_SI}. We complement 
the process with the loss for the $S$ compartment given by the operator $Q^L_{S+I\to I+I}(f_t^S,f_t^I) = Q^L_{S+I\to I+I}(f_t^S,f_t^I)(w,t)$
\[
\int_{[-1,1]} \phi(w) Q^L_{S+I\to I+I}(f_t^S,f_t^I)(w)dw = \int_{[-1,1]} \phi(w)\kappa(w,w_*)df_t^S(w)df_t^I(w_*). 
\]
Finally, the operators $Q_{SH}(f^S_t,f^H_t)$ and $Q_{IH}(f^I_t,f^H_t)$ are defined as follows
\begin{equation}\label{Kernel_Social}
\int_{[-1,1]} \phi(w) Q_{JH}(f^J_t,f^H_t)(w,t)dw 
= \int_{{[-1,1]}\times {[-1,1]}}\left\langle\phi(w^\prime)-\phi(w)\right\rangle df^S_t(w) df^I_t(w_*), 
\end{equation}
where $(w^\prime,w_*^\prime)$ has bee defined in \eqref{Def_wCC}. 


Notice the various interaction operators $Q_{S\to S},..$ appearing in the r.h.s. of the kinetic system \eqref{full.system} are in fact signed measures 
on $[-1,1]$ that we chose to denote as functions for notational convenience only. 
This system must then be understood in weak form as follows

\begin{equation}\label{full.system.weak}
\begin{split}
\frac{d}{dt} \int_{[-1,1]} \phi(w) \,df_t^I(w)
 = &  
\nu \int_{[-1,1]}  (\phi(w'_I)-\phi(w) ) \,df_t^I(w)
- \gamma \int_{[-1,1]}\phi(w)\,df^I_t(w) \\
& + \frac{1}{\varepsilon} \int_{[-1,1]\times[-1,1]} \left\langle\phi(w')-\phi(w)\right\rangle\,df^I_t(w)df^S_t(w^*) \\
& + \frac{1}{\varepsilon} \int_{[-1,1]\times[-1,1]} \left\langle\phi(w')-\phi(w)\right\rangle\,df^I_t(w)df^I_t(w^*) \\
& + \frac{1}\tau \int_{[-1,1]\times[-1,1]} \phi(w'_{S\rightarrow I}) \kappa(w,w_*) \,df_t^S(w)df_t^I(w_*),  \\
\frac{d}{dt} \int_{[-1,1]} \phi(w)\,df_t^S(w)
 = &  \nu \int_{[-1,1]}( \phi(w'_S)-\phi(w))  \,df_t^S(w)
+ \gamma \int_{[-1,1]}\phi(w'_{I\rightarrow S})\,df^I_t(w) \\
& + \frac{1}{\varepsilon} \int_{[-1,1]\times[-1,1]} \left\langle\phi(w')-\phi(w)\right\rangle\,df^S_t(w)df^S_t(w^*) \\
& + \frac{1}{\varepsilon} \int_{[-1,1]\times[-1,1]} \left\langle\phi(w')-\phi(w)\right\rangle\,df^S_t(w)df^I_t(w^*) \\
& - \frac{1}\tau \int_{[-1,1]\times[-1,1]} \phi(w) \kappa(w,w_*) \,df_t^S(w)df_t^I(w_*)   \\
& + \frac{1}\tau \int_{[-1,1]\times[-1,1]} [\phi(w'_{S\rightarrow S})-\phi(w)] (1-\kappa(w,w_*)) \,df_t^S(w)df_t^I(w_*). 
\end{split}
\end{equation}
In \eqref{full.system.weak} the parameters $\nu$, $\gamma$, $1/\varepsilon$ and $1/\tau$ are the rates at which the different processes occur, namely memory fading, recovery, social dynamic and physical encounters.

\subsection{Well-posedness}

Denote $\M([-1,1])$ the space of finite Borel measures on $[-1,1]$, $\M_+([-1,1])$ the subset of nonnegative measures, 
and $\P([-1,1])$ the subset of probability measures. 
We endow $\M([-1,1])$ with the total variation (TV) norm defined by 
$$ \|f\|_{TV} = \sup_{\phi\in C([-1,1]),\, |\phi|\le 1} \int_{[-1,1]}\phi\,df, \qquad f\in \M([-1,1]). $$



Following a similar approach as the one presented in \cite{LT}, we can re-write the system \eqref{full.system.weak} as just one equation for the probability measure 
$g_t:=f_t^I\otimes \delta_I + f_t^S\otimes \delta_S$ on $[-1,1]\times \{I,S\}$.
Each of the mechanisms described above, namely memory fading, recovery, contagion-based and pure social interaction, 
can be described by one kernel for $g_t$. 
Then the system \eqref{full.system.weak} can be written as $\p_t g_t=Q(g_t)$. 
The collision operator $Q$ satisfies some nice properties allowing to apply Bressan' techniques \cite{Bressan} 
and obtain the existence of a solution $g$, necessarily unique.
In terms of $f_t^S$ and $f_t^I$ the existence and uniqueness of $g$ reads 

\begin{thm}\label{Thm:well-posedness}
For any initial condition $f_0^S,f_0^I\in \M_+([0,1])$ such that $f_0=f_0^S+f_0^I\in \P([0,1])$ 
there exists a unique solution $f^S, f^I\in C([0,+\infty),\M_+([0,1]))\cap C^1([0,+\infty),\M([0,1]))$ 
with $f_t=f_t^S+f_t^I\in \P([0,1])$. 
Here $\M([0,1])$ and $\M_+([0,1])$ are endowed with the Total Variation norm. 
\end{thm}

\begin{proof} See Appendix \ref{Appendix:well-posedness}.
\end{proof}

\subsection{Evolution of the moments.}

In this section, we compute the evolution of the macroscopic observables of the kinetic model introduced for epidemic dynamics. In particular, from Equation~\ref{full.system.weak}, we can derive the evolution of the various macroscopic quantities. By choosing $\phi = 1$, we obtain the time evolution of the fraction of agents in each compartment

\begin{equation}\label{mass1} 
\begin{split}
\frac{d}{dt}\rho_t^I & =  -\gamma \rho_t^I  + \frac{1}\tau \iint_{[-1,1]^2}  \kappa(w,w_*) \,df_t^S(w)df_t^I(w_*) , \\
\frac{d}{dt}\rho_t^S & = \gamma \rho_t^I  - \frac{1}\tau \iint_{[-1,1]^2} \kappa(w,w_*) \,df_t^S(w)df_t^I(w_*).   
\end{split}
\end{equation} 

By direct inspection it can be noticed that, for all $t\ge0$, the total mass is conserved since $\rho_t^I+\rho_t^S = 1$. Therefore the fraction of agents at each compartment is completely determined by observing the evolution of just one compartment. In particular, under the choice of transmission rate introduced in \eqref{ContagionRate} we have 

\begin{equation}\label{Mass2} 
\frac{d}{dt}\rho_t^I =  -\gamma \rho_t^I  + \frac{1}\tau \frac{\beta}{4^\alpha} \int_{[-1,1]}(1-w)^\alpha   \,df_t^S(w)\int_{[-1,1]}(1-w)^\alpha   \,df_t^I(w). 
\end{equation} 
Therefore,  the standard SIS model is then recovered when $\alpha=0$, namely 
\begin{equation}\label{Mass_0} 
\frac{d}{dt}\rho_t^I = -\gamma \rho_t^I  + \frac\beta\tau \rho_t^I\rho_t^S. 
\end{equation} 

However, for $\alpha>0$, the equation for $\rho_t^I$ involves higher-order moments of $f_t$, leading to an expression that cannot be solved explicitly. In the next section, we introduce a closure relation by assuming that social interactions occur at high frequency; this allows us to obtain a reduced system of ODEs describing the evolution of the macroscopic quantities.

\section{Derivation of reduced-complexity systems}
\label{sect:no_selfthinking}

\subsection{Zero-diffusion case: formal derivation of the macroscopic system}  \label{Heursitic_Dirac}

In the absence of self-thinking dynamics, the pure opinion dynamic promotes consensus in the population. 
Indeed consider the systen 
\begin{align*}
\p_t f_t^I  & = Q_{IS} + Q_{II}, \\
\p_t f_t^S  & = Q_{SS}+Q_{SI}, 
\end{align*} 
where the kernels $Q_{HH'}$ are defined in \eqref{Kernel_Social} and correspond to the social 
interactions \eqref{Def_wCC} with $\sigma=0$. 
Notice that $\rho^I$ and $\rho^S$ are then constant in time. 
Also if $P$ is symmetric, i.e. $P(w,w_*)=P(w_*,w)$ for any $w,w_*\in [-1,1]$, then 
the mean opinion in the whole population $\langle w\rangle := \int w\,df_t$ is constant in time. 
The following result shows that, under some mild assumptions on $P$, individuals, regardless of their state,  tend to share the same opinion, namely the initial mean opinion $\langle w\rangle$: 

\begin{prop}\label{prop:consensus} 
Suppose that $P$ is bounded, symmetric and positive on $[-1,1]\times [-1,1]$. If $\xi < 1/\|P\|_\infty$, then 
\begin{equation}\label{SameOpinion}
f_t^I\to \rho^I \delta_{\langle w\rangle},\qquad\text{and} \qquad f_t^S\to \rho^S\delta_{\langle w\rangle} 
\end{equation}
as $t\to +\infty$, and there exists $C>0$ such that 
\begin{equation}\label{ExplicitEstimate}
 Var[f_t]\le Var[f_0] e^{-\frac{Ct}{\eps}} 
 \end{equation}
 where $Var[f_t]=\int w^2\,df_t(w) - \langle w\rangle^2$ is the variance of $f_t$.
\end{prop}

\begin{proof} See Appendix \ref{proof:prop_consensus}.
\end{proof}

Coming back to the full system \eqref{full.system.weak}, we  expect,  in view of this result, 
that in the limit $\eps\to 0$   all individuals instantaneously reach consensus 
\begin{equation}\label{ApproxConsensus}
 f_t^I \approx \rho_t^I \delta_{m_t},\qquad f_t^S \approx \rho_t^S\delta_{m_t}, 
\end{equation}
where $m_t\in [-1,1]$ is the common opinion and $\rho_t^I$, $\rho_t^S$ are the mass of $f_t^I$ and $f_t^S$. 
Our purpose now is to find a system of two equations satisfied by $(\rho_t^I, m_t)$ assuming that this approximation is valid. 

First, under this approximation,  it follows from \eqref{mass1} that 

\begin{eqnarray}
\frac{d}{dt}\rho_t^I  \approx  -\gamma \rho_t^I  + \frac{1}\tau  \rho_t^S \rho_t^I \kappa(m_t,m_t),\qquad 
\kappa(m_t,m_t) = \frac{\beta}{4^\alpha}  (1-m_t)^{2\alpha} 
\end{eqnarray}
with $\rho_t^S+\rho_t^I=1$. We can find an equation for $m_t$ as follows. Let us first rewrite the system 
\eqref{full.system.weak} as 
\begin{equation}\label{FullSYstem2}
\begin{split}
\frac{d}{dt} \int_{[-1,1]}\phi \,df_t^I 
 = &  (L[f_t^I, f_t^S],\phi) + \frac{1}{\eps} \iint_{[-1,1]^2}[\phi(w'_{IS})-\phi(w)]\,df^I_t(w)df_t(w^*) \\
 \frac{d}{dt} \int \phi \,df_t^S, 
 = &  (M[f_t^I, f_t^S],\phi) + \frac{1}{\eps} \iint_{[-1,1]^2} [\phi(w'_{IS})-\phi(w)]\,df^S_t(w)df_t(w^*) 
\end{split}
\end{equation}

where 

\begin{equation}\label{Def_L}
\begin{split}
(L[f_t^I, f_t^S],\phi) 
= &\,\nu \int_{[-1,1]} \langle \phi(w'_I)-\phi(w) \rangle \,df_t^I(w)
- \gamma \int_{[-1,1]} \phi(w)\,df^I_t(w) \\
 & + \frac{1}\tau \iint_{[-1,1]^2} \phi(w'_{S\rightarrow I}) \kappa(w,w_*) \,df_t^S(w)df_t^I(w_*),  
 \end{split}
\end{equation}

\begin{equation}\label{Def_M}
\begin{split}
(M[f_t^I, f_t^S],\phi) 
= & \, \nu \int_{[-1,1]} \langle \phi(w'_S)-\phi(w) \rangle \,df_t^S(w)
+ \gamma \int_{[-1,1]} \phi(w'_{I\rightarrow S})\,df^I_t(w) \\
& - \frac{1}\tau \iint_{[-1,1]^2} \phi(w) \kappa(w,w_*) \,df_t^S(w)df_t^I(w_*)   \\
& + \frac{1}\tau \iint_{[-1,1]^2}  [\phi(w'_{S\rightarrow S})-\phi(w)] (1-\kappa(w,w_*)) \,df_t^S(w)df_t^I(w_*). 
 \end{split}
\end{equation}
Notice that under the approximation \eqref{ApproxConsensus}  the integral in the r.h.s. of \eqref{FullSYstem2} vanishes. 
Hence, at time $t + \Delta t$, $\Delta t>0$, the transfer of mass between the S and I compartments and the movements of mass within each compartment  encoded in $L$ and $M$  slightly perturb the Dirac masses at $m_t$ adding new Dirac masses.  This result in the new measures $\tilde f^S_{t+\Delta t}$ and $\tilde f^I_{t+\Delta t}$ given by
$$ \tilde f^I_{t+\Delta t} = \rho_t^I \delta_{m_t} +  \Delta t L[\rho_t^I\delta_{m_t}, \rho_t^S\delta_{m_t}] $$
and 
$$ \tilde f^S_{t+\Delta t} = \rho_t^S \delta_{m_t} +  \Delta t M[\rho_t^I\delta_{m_t}, \rho_t^S\delta_{m_t}]. $$
Then, due the very high frequency of the pure interaction kernels, 
opinions in both compartments converge to $m(t+\Delta t)$, the mean opinion 
of $\tilde f_{t+\Delta t}:=\tilde f^S_{t+\Delta t}+\tilde f^I_{t+\Delta t}$. 
Summing the two previous equations gives 
$$ \tilde f_{t+\Delta t} = \delta_{m_t} +  \Delta t(L[\rho_t^I\delta_{m_t}, \rho_t^S\delta_{m_t}] + M[\rho_t^I\delta_{m_t}, \rho_t^S\delta_{m_t}]). $$ 
Taking the mean, dividing by $\Delta t>0$  we get 
$$ \dfrac{d}{dt}m_t = (L[\rho_t^I\delta_{m_t}, \rho_t^S\delta_{m_t}] + M[\rho_t^I\delta_{m_t}, \rho_t^S\delta_{m_t}], w). $$ 

The previous formal reasoning produces the following reduced system of ODEs for $(\rho_t^I,m_t)$: 
\begin{equation}\label{ReducedODE}
\begin{split}
\frac{d}{dt}\rho_t^I   = &  -\gamma \rho_t^I  + \frac{1}\tau  \rho_t^S \rho_t^I \kappa(m_t,m_t) \\
\dfrac{d}{dt}m_t = & (L[\rho_t^I\delta_{m_t}, \rho_t^S\delta_{m_t}] + M[\rho_t^I\delta_{m_t}, \rho_t^S\delta_{m_t}], w). 
\end{split}
\end{equation}
with $\rho_t^S = 1- \rho_t^I$. Using the definition of the microscopic interaction scheme introduce in the previous section we obtain the following system of ODEs that describe the evolution of the macroscopic quantities. 

\begin{equation}\label{Approx_Dirac}
\begin{split}
\frac{d}{dt}\rho_t^I   = &  -\gamma \rho_t^I  + \frac{1}\tau  \rho_t^S \rho_t^I \kappa(m_t,m_t) \\
 \frac{d}{dt} m_t =&  \gamma \rho_t^I h (-1-m_t) + \nu h (\rho_t^I  + \rho_t^S)(-1-m_t) \\
& + \frac1\tau \kappa(m_t,m_t)\rho_t^I\rho_t^S h(1-m_t)+ \frac1\tau h(-1-m_t)  (1-\kappa(m_t,m_t))\rho_t^I\rho_t^S, 
\end{split}
\end{equation}
where $\rho_t^I+\rho_t^S=1$ for all $t\ge0$. 

\subsubsection{Rigorous derivation of the macroscopic model}

In view of \eqref{FullSYstem2}, we consider a general system of the form 
\begin{equation}\label{FullSystem3}
\begin{split}
\partial_t f_t = L[f_t,g_t] + \frac1\varepsilon Q[f_t,h_t],\\ 
\partial_t g_t = M[f_t,g_t] + \frac1\varepsilon Q[g_t,h_t],
\end{split}
\end{equation} 
where $f_t,g_t$ are non-negative measures on $[-1,1[$ such that $h_t:=f_t+g_t$ is a probability measure, 
$L,M:\mathcal{A}\to \M([-1,1])$ where 
$\mathcal{A}:=\{(f,g)\in \M_+([-1,1])\times \M_+([-1,1]) \text{ s.t. } f+g\in \P([-1,1])\}$, 
and $Q:\M_+([-1,1])\times \P([-1,1])\to \M([-1,1])$. 
Under some assumption on $L,M,Q$ precisely stated in Appendix \ref{Appendix_Section_DimReduction} and satisfied in particular 
by the tendency-to-compromise kernel
$$ (Q[f,h],\phi) = \iint [\phi(x+P(x,x_*)(x_*-x))-\phi(x)]\,df(x)dh(x_*),$$
it can be proved that $(f_t^\varepsilon,g_t^\varepsilon)$ is close to a pair of Dirac mass of the form 
$(r_t\delta_{m_t}, (1-r_t)\delta_{m_t})$ in the limit $\varepsilon\to 0$.

Indeed, assuming that $Q$ is linear in the 1st argument, we can add both equation to  obtain
\begin{equation}\label{Equ_h}
 \p_th_t = L[f_t,g_t] + M[f_t,g_t] + \frac1\eps Q[h_t,h_t]. 
\end{equation}
Due to a kind of contractivity assumption on $Q$, 
the solution $h_t^\varepsilon$ of \eqref{Equ_h} will rapidly approach (on the time time scale $t/\varepsilon$) 
the slow manifold $\{\delta_m,\,m\in\mathbb{R}\}$ and then stay close to it. 
In fact we will prove that 
\begin{equation*}
    W_2(h_t^\varepsilon,\delta_{m^\varepsilon(t)}) 
    \le C(e^{-\lambda t/(2\varepsilon)}  + \sqrt{\varepsilon}) 
\end{equation*}
for some $C,\lambda>0$, and where $m^\varepsilon(t)$ is the mean of $h_t$ and $W_2$ is the quadratic Monge-Kantorovich or Wasserstein distance. 
Since $h=f+g$, we thus expect intuitively that $f_t\simeq \rho_t \delta_{m_t}$ 
and $g_t\simeq (1-\rho_t) \delta_{m_t}$ with $\rho_t$ the total mass of $f_t$. 
A heuristic reasonning as before then leads us to the following limit system: 
\begin{equation}\label{LimitEq20}
\begin{split}
\frac{d}{dt} \rho_t & = (L[\rho_t \delta_{m_t},(1-\rho_t) \delta_{m_t}], 1),\\ 
\frac{d}{dt} m_t & = (L[\rho_t \delta_{m_t},(1-\rho_t) \delta_{m_t}] 
+ M[\rho_t \delta_{m_t},(1-\rho_t) \delta_{m_t}], x).
\end{split} 
\end{equation} 
This approximation can be justified assuming Lipshitz continuity w.r.t. $(f,g)$ of $(L[f,g],1)$ and $(L[f,g],x)$ (same for $M$) 
w.r.t. the distance $\widetilde{W}_2$ on $\M_+([-1,1])$ defined as  
\begin{equation}\label{Appendix_Def_tildeW2} 
 \widetilde{W}_2(f,f') := W_1(\tilde f,\tilde f') + |\rho-\rho'|, \qquad f,f'\in\M_+([-1,1]),
\end{equation}
where $\rho$, $\rho'$ are the total mass of $f,f'$, and $\tilde f:=f/\rho$, $\tilde f':=f'/\rho'$, see \cite{CPSV,PiccoliRossi,Villani09}

We can then prove the following result

\begin{thm}\label{Thm_Approx} 
Consider an initial condition $(f_0,g_0)\in\mathcal{A}$ and suppose that \eqref{FullSystem3} 
has a unique solution $(f^\varepsilon,g^\varepsilon)$ for  $\varepsilon>0$ small. 
Denote $(\rho_t,m_t)$ the solution of \eqref{LimitEq20} with initial 
conditions $\rho_0=(f_0,1)$ and $m_0=(f_0+g_0,x)$.  
Then, under some specific assumptions on $L,M,Q$ stated in Theorem \ref{Appendix_Thm_2Eq},  
for any $t$ such that $\rho_s^\varepsilon\in (0,1)$, $0\le s\le t$, there holds 
\begin{equation}\label{Appendix_result_2Eq}
 \widetilde{W}_2(f_t^\varepsilon, \rho_t\delta_{m_t})
+ \widetilde{W}_2(g_t^\varepsilon, (1-\rho_t)\delta_{m_t})
\le C'( \sqrt{\varepsilon} + \varepsilon e^{C t}).
\end{equation}
for some $C,C'>0$. 
\end{thm}

\noindent A precise statement and proof of the Theorem are given in Appendix \ref{Appendix_Section_DimReduction}.

\subsection{Macroscopic SIS model in the presence of self-thinking dynamics}\label{sect:FP}

The balance between compromise and self-thinking has been exploited to obtain a mean-field description of the kinetic opinion dynamics defined in \eqref{Def_wCC}, see \cite{ACDZ,C,CDTZ,PT,T}. As shown in \cite{T,Z}, by posing the correct scaling known as the quasi-invariant regime, where we consider that the interactions between the agents occur at high frequencies and their post-interactions opinion are slightly modified, we obtain a reduced complexity model of the form of a Fokker-Planck equation where the study of the asymptotic properties become easier to analyze. In particular, under the scaling $\xi\to \epsilon \xi$ and $\sigma^2 \to \epsilon\sigma^2$, in the limit $\epsilon\to 0^+$, the operators $Q_{SH}(\cdot,\cdot)$ and $Q_{IH}(\cdot,\cdot)$ correspondong to the social interaction \eqref{Def_wCC} with $P\equiv 1$ simplify into the following 

\begin{equation}\label{FP}
    \tilde Q_{KH}(f_t^H)=\partial_{w} [ \xi (w - m_t)f^H_t(w,t)+\frac{\sigma^2}{2}\partial_{w}(D^2(w)f_t^H(w,t))], \qquad K,H\in \{S,I\},
\end{equation}
complemented by no-flux boundary conditions and where $m_t$ is the mean of $f_t^H$.
%

As argued before, in the absence of self-thinking forces, the evolution of opinions is dictated solely by an aggregation dynamics, concentrating the opinions around the mean value $m_t$. 
On the other hand, for $\sigma^2>0$ and $D(w,w_*)=\sqrt{1-w^2}$, 
the equilibrium state of the Fokker-Planck equation \eqref{FP} satisfies
\begin{equation*}
    \xi (w - m)S[m](w)+\frac{\sigma^2}{2}\partial_{w}(D^2(w)S[m](w) = 0
\end{equation*}
and the steady state $S$ is then a Beta-type distribution
\begin{equation}\label{eq.state}
S[m](w)\;=\; 
\frac{(1+w)^{\tfrac{1+m}{\kappa}-1}\,(1-w)^{\tfrac{1-m}{\kappa}-1}}
     {2^{\tfrac{2}{\kappa}-1}\, B\!\left(\tfrac{1+m}{\kappa},\, \tfrac{1-m}{\kappa}\right)},
\qquad 
\kappa = \frac{\sigma^2}{\xi},
\end{equation}
where $B$ is the Beta function. 
From Figure~\ref{fig:steady_state} it can be seen that the ratio $\kappa$ between the aggregation tendency and self-thinking forces gives rise to two different opinion structures: when $\kappa > 1$, the self-thinking forces dominate the evolution of opinion and we obtain a polarized  distribution of opinion, while if $\kappa < 1$, the compromise propensity dictates the evolution and we can observe the emergence of consensus formation.

\begin{figure}
    \centering
    \includegraphics[width=\linewidth]{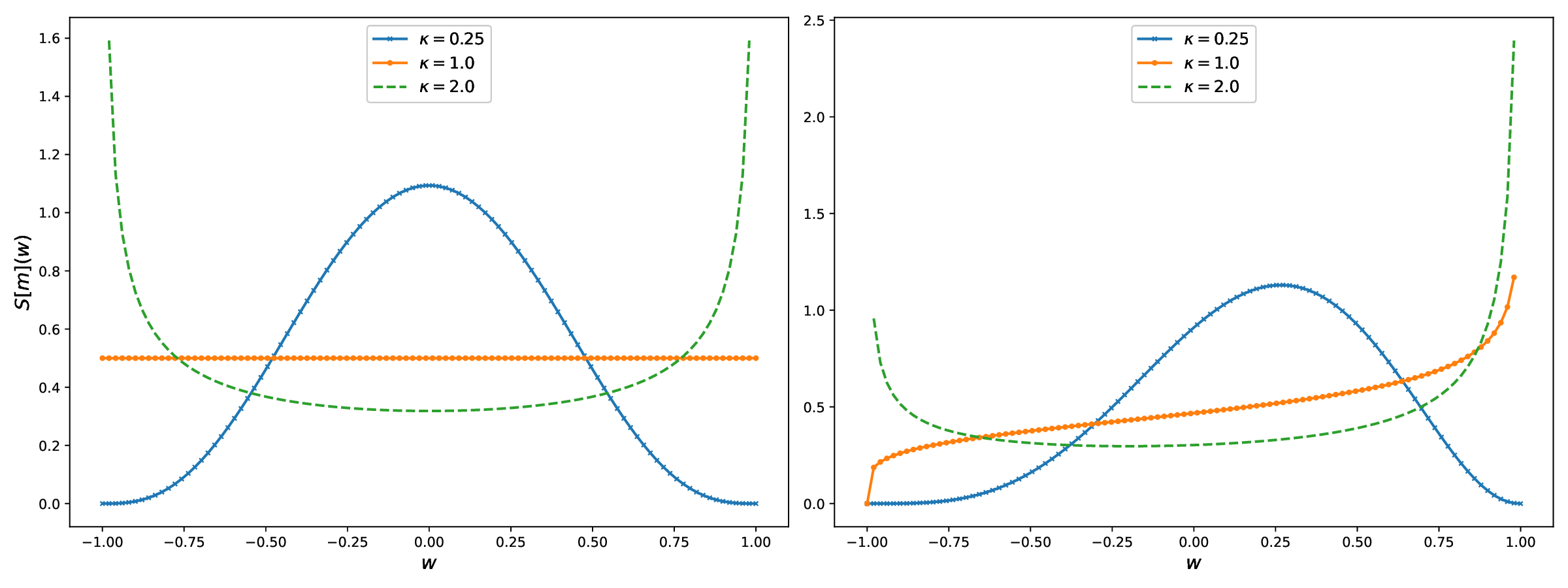}
    \caption{Equilibrium state \ref{eq.state} for different choices of the distribution parameter. When the self-thinking forces dominate the dynamics, we can observe a polarized distribution of opinions, while consensus formation emerges when the aggregation tendency surpasses the self-thinking effect.}
    \label{fig:steady_state}
\end{figure}

We now consider the presence of self-thinking dynamics which we include by adding a diffusion term as shown in scheme~\ref{Def_wCC}.  As presented above, the ratio between the consensus tendencies and self-thinking forces give rise to different opinion structures that can affect the evolution of the macroscopic quantities of our epidemic. Similarly as before, we consider that the rate of social interactions occur at high frequencies and thus the opinion distribution is characterize by a Beta-type distribution shown in Equation~\ref{eq.state}. In the presence of self-thinking forces our closure relation adopts the following form:

\begin{equation}\label{ApproximationGrazing}
 f_t^C\simeq \rho_t^CS[m_t] \qquad C\in\{S,I\},  
 \end{equation}
where $m_t$ is the mean opinion of $f_t=f_t^S+f_t^I$. Considering the full system which can be written as 
\begin{equation}\label{Grazing2}
\begin{aligned}
\partial_t f_t^I &= L[f_t^I, f_t^S] + \tfrac{1}{\varepsilon}\,\mathcal{L}[f_t^I, f_t], \\
\partial_t f_t^S &= M[f_t^I, f_t^S] + \tfrac{1}{\varepsilon}\,\mathcal{L}[f_t^S, f_t].
\end{aligned}
\end{equation}
where $L$ and $M$ are defined in \eqref{Def_L} and \eqref{Def_M}, 
and the operator $\mathcal{L}[f_t^C, f_t]$ is given by the Fokker-Planck operator \eqref{FP}. First, since $\mathcal{L}$ preserves the mass,  taking $\phi\equiv 1$ as test-function gives 
$$ \p_t \rho_t^I = (L[f_t^I, f_t^S], 1) 
= -\gamma \rho_t^I + \frac{1}\tau \iint \kappa(w,w_*) \,df_t^S(w)df_t^I(w_*).   
$$
Under the approximation \eqref{ApproximationGrazing}, we thus obtain 
\begin{equation}\label{Limit_Noise_ODE_rho}
 \p_t \rho_t^I 
= -\gamma \rho_t^I + \frac{1}\tau \rho_t^I \rho_t^S  \iint \kappa(w,w_*) S[m_t](w)S[m_t](w_*) \,dwdw_*. 
\end{equation}
A formal derivation of the evolution of $m_t$ follows through the argument in Section \ref{Heursitic_Dirac}. At time $t + \Delta t$, $\Delta t>0$,  we have
\[
 f^I_{t+\Delta t} \approx f_t^I + \Delta t\left( L[f_t^I, f_t^S] + \frac{1}{\eps}\mathcal{L}[f_t^I, f_t]\right) 
 \]
i,e., under the approximation \eqref{ApproximationGrazing} and recalling that $\mathcal{L}[S[m_t], S[m_t]]=0$, 
\[
 f^I_{t+\Delta t} \approx f_t^I + \Delta t \left(L[\rho_t^I S[m_t], \rho_t^S S[m_t]]\right).  
 \]
The same equation holds for $f_t^S$ with $M$ in place of $L$. Summing these two equations gives 
\[
 f_{t+\Delta t} \approx f_t + \Delta t(L+M)[\rho_t^I S[m_t], \rho_t^S S[m_t]].  
 \] 
Taking $\phi(w)=w$ as test-function gives 
\[
 m_{t+\Delta t} \approx m_t + \Delta t((L+M)[\rho_t^I S[m_t], \rho_t^S S[m_t]], w).  
 \] 
Then, in the limit $\Delta t\to 0^+$ we get
\begin{equation}\label{Limit_Noise_ODE_m}
 \frac{d}{dt} m_t = (L+M)[\rho_t^I S[m_t], \rho_t^S S[m_t], w). 
\end{equation}
In conclusion, the solution $(f_t^I, f_t^S)$ of \eqref{Grazing2} should be well-approximated in the limit $\eps\to 0$
by $(\rho_t^S S[m_t], \rho_t^I S[m_t])$ with 
\begin{equation}\label{HighFreqGrazing}
\begin{split}
\frac{d}{dt} \rho_t^I 
& = -\gamma \rho_t^I + \frac{1}\tau \rho_t^I \rho_t^S  \iint \kappa(w,w_*) S[m_t](w)S[m_t](w_*) \,dwdw_*, \\ 
\frac{d}{dt} m_t & = (L[\rho_t^I S[m_t], \rho_t^S S[m_t]], w) + (M[\rho_t^I S[m_t], \rho_t^S S[m_t]], w). 
\end{split} 
\end{equation}
Notice that, in the limit $\sigma/\gamma\to 0$, we formally have $S[m]\to \delta_m$ and \eqref{HighFreqGrazing} becomes 
\eqref{ReducedODE}, the reduced system of ODEs in absence of self-thinking.  Recalling the definition \eqref{Def_L} and \eqref{Def_M} of $L$ and $M$, we have 
\begin{equation*}
\begin{split}
 (L[\rho_t^I S[m_t], \rho_t^S S[m_t]], w)
= & \, \nu \rho^I_t\int (w'_I-w) S[m_t](w) \,dw 
- \gamma \rho^I_t m_t  \\
 & + \frac{1}\tau \rho^I_t \rho^S_t \iint w'_{S\rightarrow I} \kappa(w,w_*) S[m_t](w) S[m_t](w_*) \,dw dw_*
 \end{split}
\end{equation*}
and 
\begin{equation*}
\begin{split}
(M[\rho_t^I S[m_t], \rho_t^S S[m_t]], w)
= &\, \nu \rho_t^S\int (w'_S-w)  S[m_t](w)\,dw
+ \gamma \rho_t^I \int w'_{I\rightarrow S} S[m_t](w)\,dw \\
 &- \frac{1}\tau \rho^I_t \rho^S_t \iint w \kappa(w,w_*) S[m_t](w) S[m_t](w_*) \,dw dw_* \\
 &+ \frac{1}\tau \rho^I_t \rho^S_t \iint  (w'_{S\rightarrow S}-w) (1-\kappa(w,w_*)) S[m_t](w) S[m_t](w_*) \,dw dw_*
 \end{split}
\end{equation*}
Thus we may write 
\begin{equation}\label{OpLM-special}
\begin{aligned}
\frac{1}{h}\,((L+M)[\rho_t^I S[m_t], \rho_t^S S[m_t]], w)
&= -\Bigl(\nu + \gamma \rho_t^I + \tfrac{1}{\tau}\rho^I_t \rho^S_t\Bigr)(1+m_t) \\
&\quad + \frac{2}{\tau}\,\rho^I_t \rho^S_t \iint \kappa(w,w_*) S[m_t](w) S[m_t](w_*) \,dw\,dw_* .
\end{aligned}
\end{equation}
and in a similar fashion as before we may describe the evolution of the macroscopic quantities by the following system of ODEs.

\begin{equation}\label{HighFreqGrazing2}
\begin{aligned}
\frac{d}{dt} \rho_t^I 
&= -\gamma \rho_t^I 
+ \frac{1}{\tau}\,\rho_t^I \rho_t^S \iint \kappa(w,w_*) S[m_t](w) S[m_t](w_*) \,dw\,dw_*, \\[6pt]
\frac{1}{h}\frac{d}{dt} m_t 
&= -\Bigl(\nu + \gamma \rho_t^I + \tfrac{1}{\tau}\rho^I_t \rho^S_t\Bigr)(1+m_t) \\
&\quad + \frac{2}{\tau}\,\rho^I_t \rho^S_t \iint \kappa(w,w_*) S[m_t](w) S[m_t](w_*) \,dw\,dw_* .
\end{aligned}
\end{equation}
We remark that, in the limit $\sigma/\gamma\to 0$, we consistently recover the model defined in \eqref{Approx_Dirac}.

\section{Numerical Results}\label{sect:num}

In this section we present different numerical examples to show the consistency and validity of the different approaches proposed in this work. We will show how the self-thinking parameter impacts the opinion formation process rendering different types of equilibrium states.
The coupled dynamics decreases the asymptotic infection level and introduces new dynamical behaviors.

From the numerical perspective we consider the Direct Simulation Monte Carlo (DSMC) method to show how, in the quasi-invariant regime, the asymptotic distribution of the Boltzmann-type model describing the social interactions is consistent with the Fokker-Planck model presented in Section~\ref{sect:FP}. Moreover, by adopting a time-splitting strategy, we demonstrate the consistency between the full Boltzmann model and the evolution of the corresponding macroscopic quantities, both in the absence and in the presence of self-thinking dynamics, as described by Equations~\ref{HighFreqGrazing2} and \ref{Approx_Dirac}. For a detailed description of these methods we point the interested reader to~\cite{PR}. \

Lastly, we explore the influence of opinion dynamics on epidemic evolution, highlighting how the different interaction mechanisms can generate different epidemic patterns. We illustrate that the coupling between opinion formation and disease spread can substantially alter the overall system dynamics.
\begin{figure}
    \centering

    \begin{subfigure}{0.48\textwidth}
        \centering
        \includegraphics[width=\linewidth]{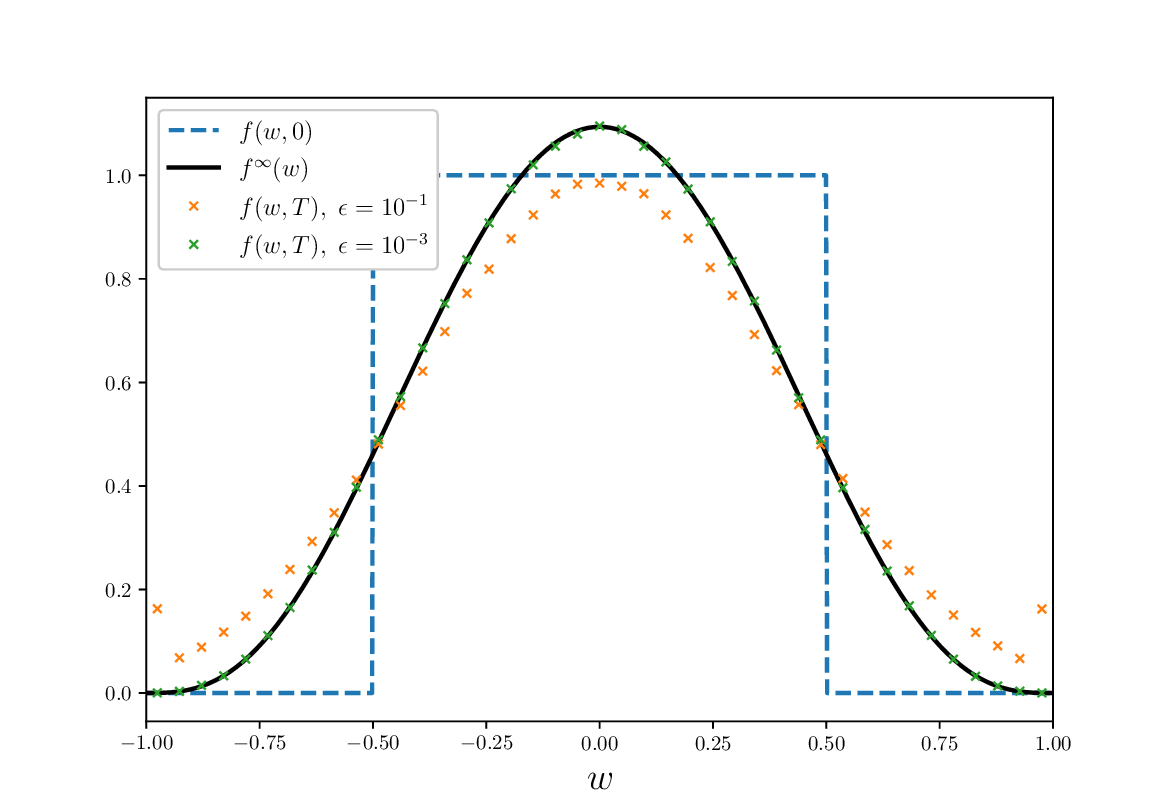}
        \caption{}
    \end{subfigure}
    \hfill
    \begin{subfigure}{0.48\textwidth}
        \centering
        \includegraphics[width=\linewidth]{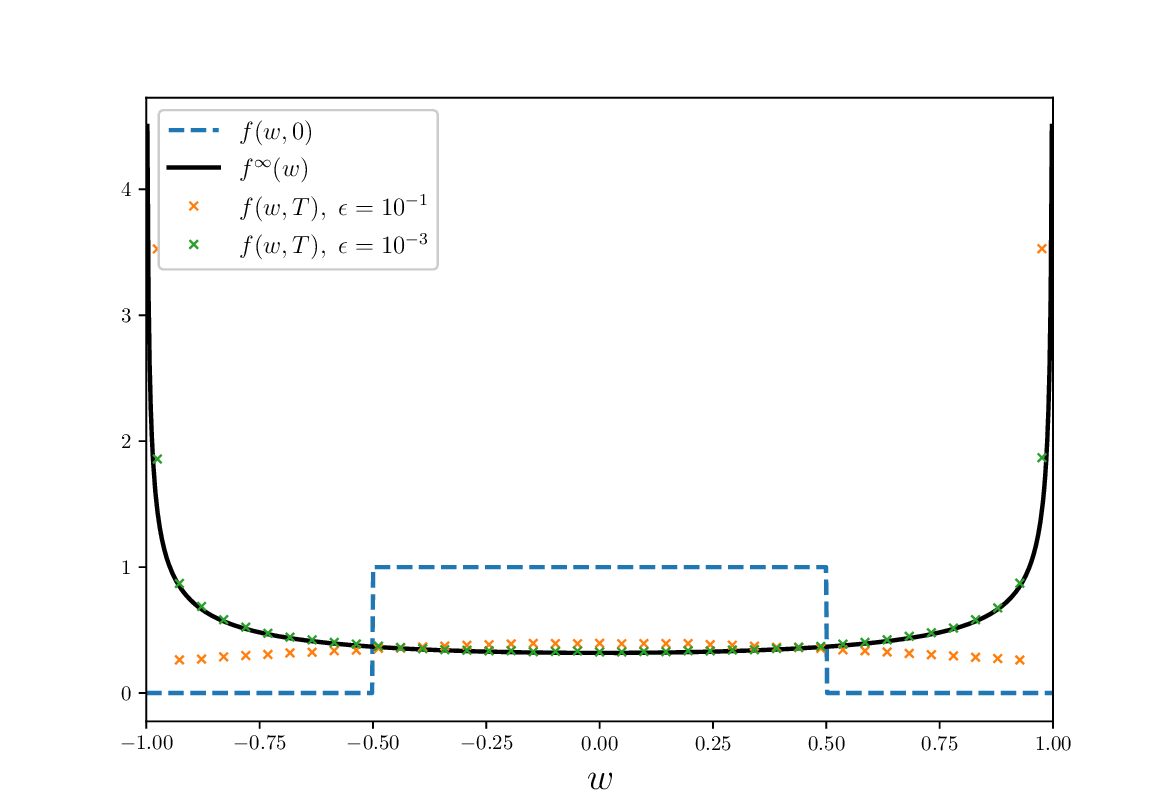}
        \caption{}
    \end{subfigure}

    \begin{subfigure}{0.48\textwidth}
        \centering
        \includegraphics[width=\linewidth]{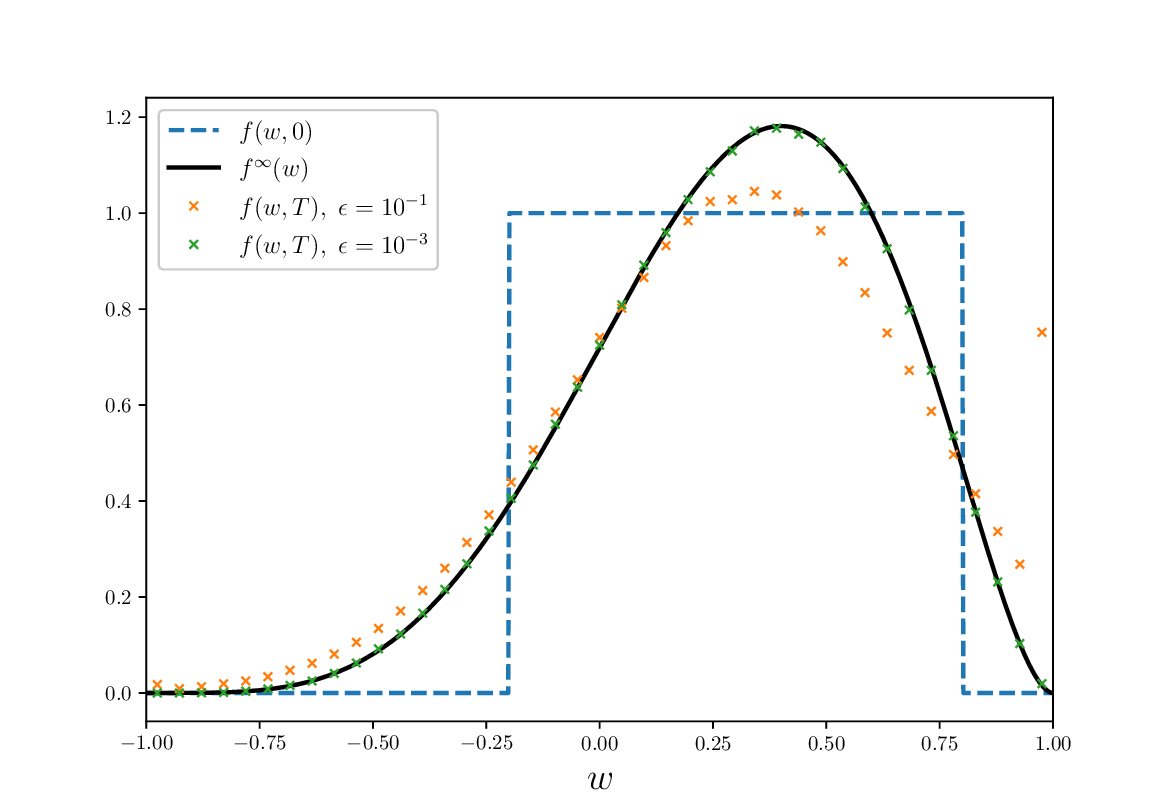}
        \caption{}
    \end{subfigure}
    \hfill
    \begin{subfigure}{0.48\textwidth}
        \centering
        \includegraphics[width=\linewidth]{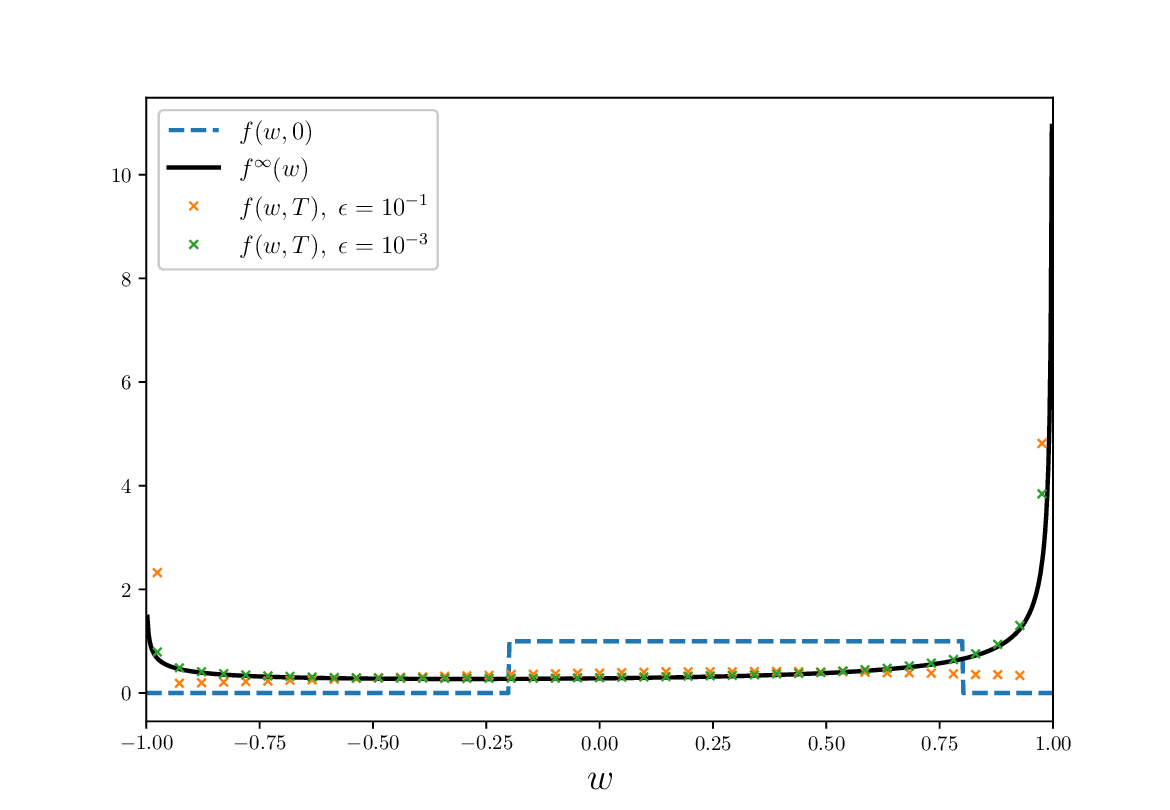}
        \caption{}
    \end{subfigure}

    \vspace{1em}

    \begin{subfigure}{0.6\textwidth}
        \centering
        \includegraphics[width=\linewidth]{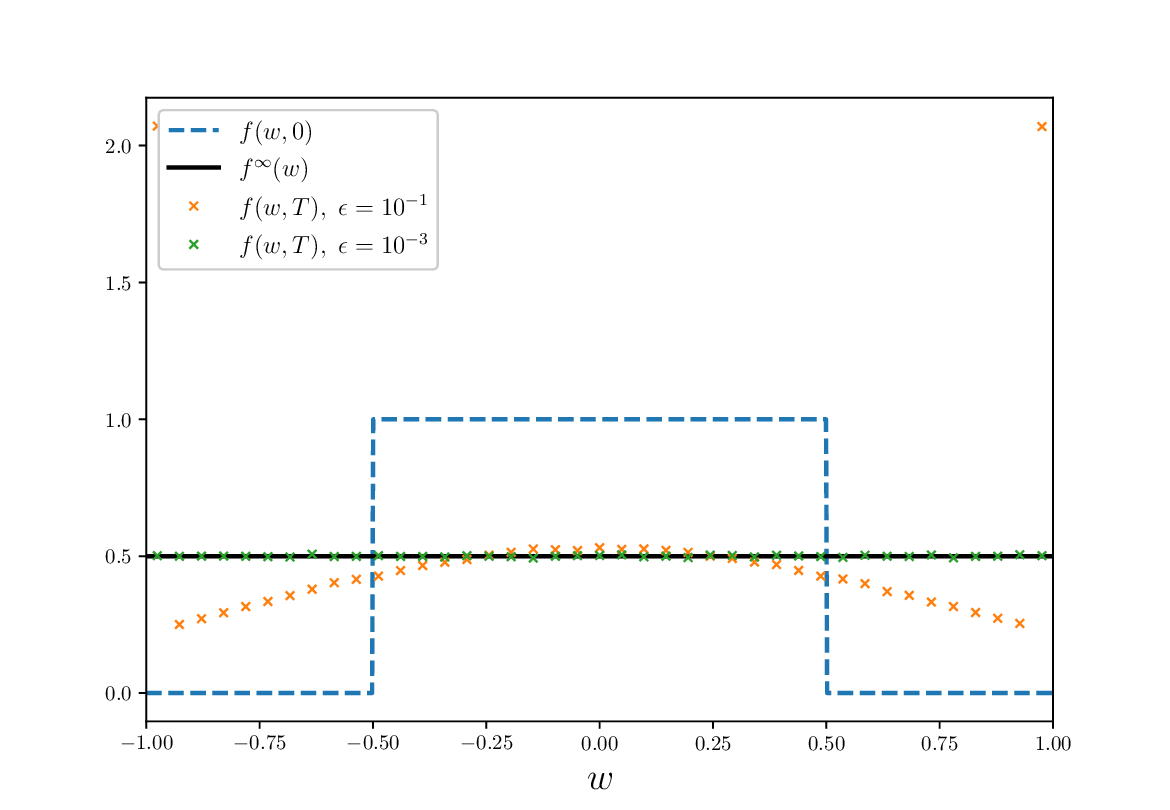}
        \caption{}
    \end{subfigure}

    \caption{Comparison between the DSMC solution of the Boltzmann-type model and the Fokker-Planck solution in the quasi-invariant regime. We consider different combinations of initial distributions and the self-thinking parameter. In particular we have in (a): $m(0)=0$, $\sigma^2=0.25$; in (b): $m(0)=0$, $\sigma^2=2$; in (c): $m(0)=0.3$, $\sigma^2=0.25$; in (d): $m(0)=0.3$, $\sigma^2=2$; in (e): $m(0)=0$, $\sigma^2=1$. We implemented the DSMC scheme with $N=10^{6}$ agents, a final time $T=5$ and $\epsilon = 10^{-1},10^{-3}$.}
    \label{fig:beta_particle}
\end{figure}

\begin{figure}
    \centering
    \includegraphics[width=0.49\textwidth]{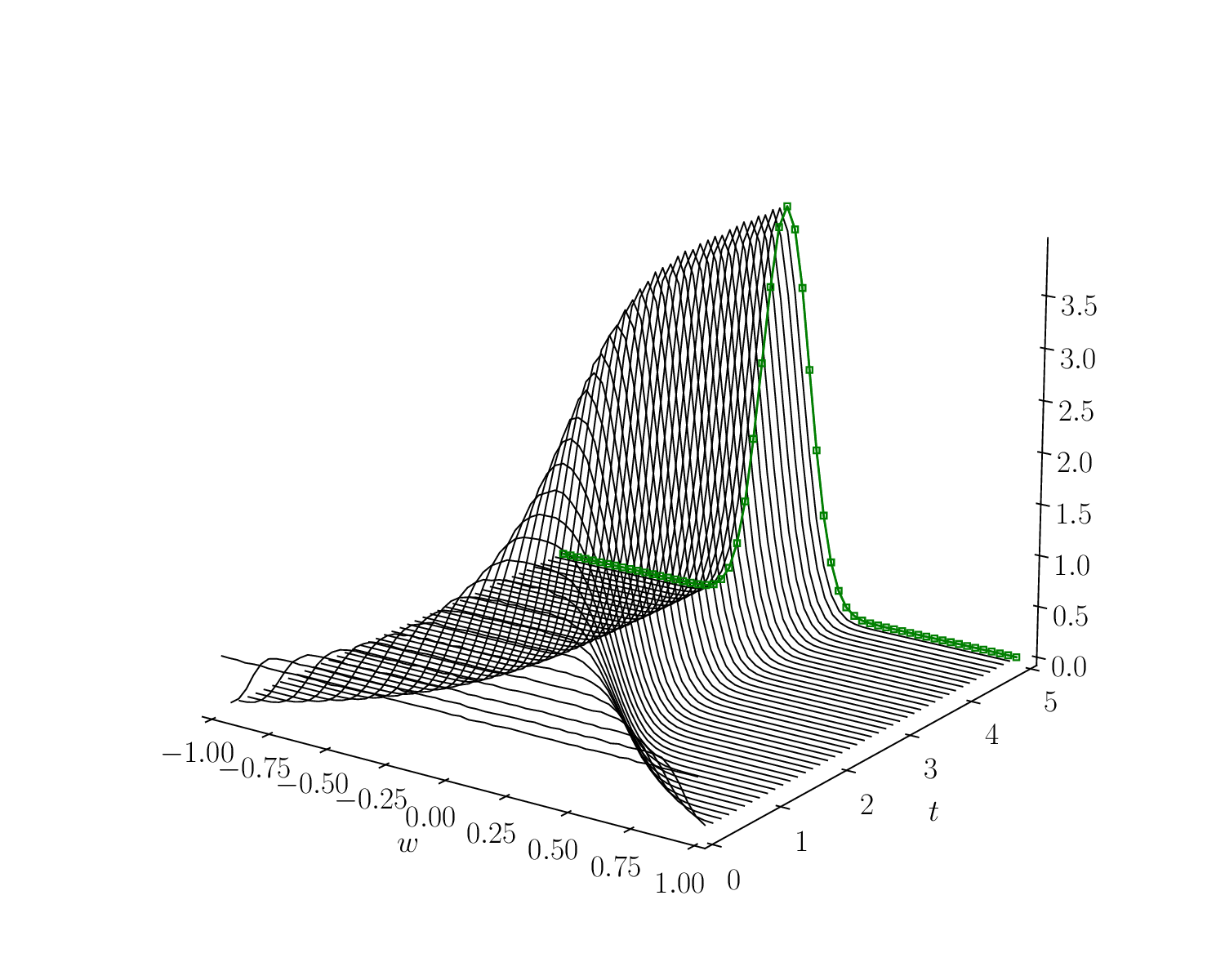}
    \hfill
    \includegraphics[width=0.49\textwidth]{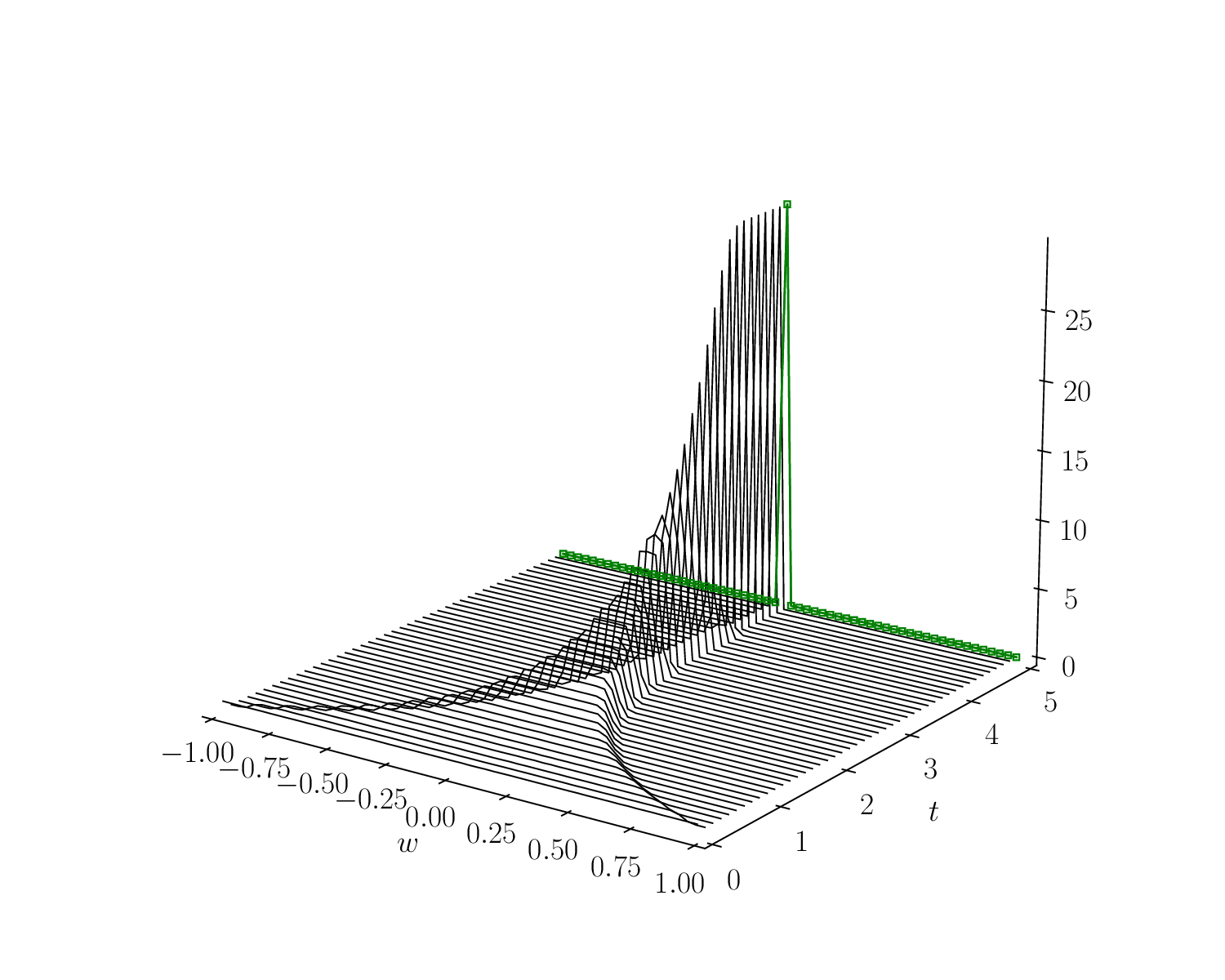}

    \caption{Implementation of the DSMC scheme in the limit where the self-thinking forces vanish. We consider as initial distribution a uniform distribution in the interval $[-1,1]$ and we take $\sigma^{2} = 10^{-1}$ (left) and $\sigma^{2} = 10^{-4}$ (right). We implemented the DSMC scheme with $N=10^{6}$ agents, a final time $T=5$ and $\epsilon = 10^{-3}$. In the limit $\sigma^{2} \to 0$ we obtain global consensus and all agents share the same opinion giving place to the formation of a singularity.}
    \label{fig:singularity}
\end{figure}

\subsection{Asymptotic behavior of the Kinetic Opinion Model}

In this first numerical test we focus solely on the opinion dynamics in the absence and presence of self-thinking forces. We test the consistency between the Boltzmann-type model given by operator $Q[f,f_{*}]$ shown in Equation~\ref{FullSystem3} and the Fokker-Planck model under the quasi-invariant regime. Under this regime, when $\epsilon \rightarrow 0^{+}$, we expect the asymptotic distribution of the Boltzmann model to converge to the beta-type distribution. We start by rewriting the Boltzmann model in the following form:

\begin{equation}
    \partial_t f_{J}(w,t) = Q^{+}(f_{J},f_{J})(w,t) - f_{J}(w,t),
\end{equation}
where

\begin{equation}
    Q^{+}(f_{J},f_{J})(w,t) = \left\langle \int_{-1}^{1} \frac{1}{J} f_{J}(w',t) f_{J}(w'_{*},t) dw_{*} \right\rangle,   
\end{equation}
being $J$  the jacobian of the transformation between the pre-interaction and post-interaction opinions $('w,'w_{*})\rightarrow('w,'w_{*})$. 
To compute the long-time solution of the Boltzmann model under the quasi-invariant regime we take $\epsilon = 10^{-1},10^{-3}$ and consider $N=10^{6}$ agents with a final time $T=5$ and $\Delta t = \epsilon$. As initial distributions we consider:

\begin{equation}
    f^{(1)}(w,0) = \begin{cases}
    1 & w \in [-0.5,0.5] \\
    0 & w \notin [-0.5,0.5]
    \end{cases}
\end{equation}
and 

\begin{equation}
    f^{(2)}(w,0) = \begin{cases}
    1 & w \in [-0.2,0.8] \\
    0 & w \notin [-0.2,0.8]
    \end{cases}
\end{equation}
such that $m^{(1)}_{J}(0)=0$ and $m^{(2)}_{J}(0)=0.3$ respectively. In Figure~\ref{fig:beta_particle} we show the asymptotic distribution of opinions for both initial distributions and  the different values of $\epsilon$. It can be seen that, in the limit $\epsilon \to 0^{+}$, the distribution of opinions converges toward a beta-type distribution. Moreover, it can be observed that different magnitudes of the self-thinking parameter give rise to different opinion structures. In particular, for $\sigma^{2} < 1$ the distribution is characterized by consensus around the mean opinion. For $\sigma^{2} = 1$ the distribution is uniform, whereas for $\sigma^{2} > 1$ the distribution becomes polarized, with opinions concentrating at the extremes. \

In the absence of self-thinking forces, the long time distribution is characterised by a total consensus where all agents share the same opinion. Mathematically this is represented by the formation of a Dirac delta function centered at the inital mean opinion. We consider as inital condition:

\begin{equation}
    f^{(3)}(w,0) = \begin{cases}
    1/2 & w \in [-1,1] \\
    0 & w \notin [-1,1]
    \end{cases}
\end{equation}
and obtain the long time solution of the Boltzmann-type model for decreasing values of the self-thinking parameter $\sigma^{2}$. The rest of the parameters are taken as in the previous test. The results are shown in Figure~\ref{fig:singularity}. It can be seen that in the limit when $\sigma^{2} \rightarrow 0$ we observe the emergence of the singularity.

\subsection{Consistency with the evolution of macroscopic quantities}

In this section, we demonstrate the consistency between the evolution of the number of infected agents and the population’s mean opinion as predicted by the Boltzmann-type model~\ref{full.system.weak} and the corresponding macroscopic systems, both in the absence and in the presence of self-thinking forces given by (\ref{Approx_Dirac})-(\ref{HighFreqGrazing2}). As shown in previous sections we consider the following closure relation $f_\mathcal{J}(t,w) = \rho_\mathcal{J}(t)f^{q}(t,w)$ where the form of $f^{q}(t,w)$ depends on the value of the self-thinking parameter $\sigma^2$. 
We consider a time discretisation of the interval $[0,5]$ where the size of each time interval is $\epsilon = \Delta t \geq 0$. Following the same strategy as in~\cite{ACDZ}, we introduce a time-splitting strategy between the opinion dynamics step where we first solve the following interactions such that $f^{*}_\mathcal{J}(x) = \mathcal{I}_{\Delta t}(f_S^*,f^*_I)$:

\begin{equation}
\left\{
\begin{aligned}
    \partial_tf^*_\mathcal{S} & = \frac{1}{\tau_S} \Big( Q_{SS} + Q_{II} \Big) + \frac{1}{\tau_P} \Big (Q_{S+I\rightarrow S+I} + Q_{S+I\rightarrow I+I} \Big) \\
    \partial_tf^*_\mathcal{I} & = \frac{1}{\tau_S} \Big( Q_{SS} + Q_{II} \Big) + \frac{1}{\tau_P} Q_{S+I\rightarrow I+I} 
\end{aligned}
\right.
\end{equation}
where $\tau_S$ is the rate of social interactions and $\tau_P$ the rate of physical interactions. Following we perform the epidemiological step $f^{**}_\mathcal{J}(x) = \mathcal{E}_{\Delta t}(f_\mathcal{J}^{**})$.

\begin{equation}
\left\{
\begin{aligned}
    \partial_t f^{**}_\mathcal{I} &= - \gamma f^{**}_\mathcal{I} \\
    \partial_t f^{**}_\mathcal{S} &= \gamma f^{**}_\mathcal{S}
\end{aligned}
\right.
\end{equation}
Therefore the solution at time $t^{n+1}$ is given by 

\begin{equation}
    f_\mathcal{J}^{n+1} = \mathcal{E}_{\Delta t} (\mathcal{I}_{\Delta t}(f_S^n,f_I^n))
\end{equation}

Since we are interested in the evolution of opinions for the total population we compute $f(w,t)=f_S(w,t)+f_I(w,t)$. The computational details of this methods are detailed in Algorithm~\ref{algo:DSMC}. It is important to highlight that this approach is a first-order splitting strategy. Higher order methods, as the Strang splitting technique, could also be considered.

\begin{algorithm}
\caption{Particle-based time-splitting Monte Carlo scheme.}
\label{algo:DSMC}
\begin{algorithmic}[1]
\State Fix $N$ and initialize the infection states $s_i \sim \text{Bernoulli}(\rho_{\text{inf}})$.
\State Define the number of time steps as $n_{\text{steps}} = T / \epsilon$.
\State Sample the number of agents participating in a social interaction: $n_{\text{social}} = \left\lfloor N \cdot \frac{\epsilon}{\tau_S} \right\rfloor$.
\State Sample the number of agents participating in a physical interaction: $n_{\text{physical}} = \left\lfloor N \cdot \frac{\epsilon}{\tau_P} \right\rfloor$.

\For{$t = 1$ to $n_{\text{steps}}$}
    \State \textbf{Social interactions:} Sample $N_s = \text{round}(n_{\text{social}} / 2)$ pairs of agents $(i,j)$ without repetition.    
    \ForAll{pairs $(i, j)$}
        \State Sample $\zeta_{i}^{n}, \zeta_{j}^{n}$ from a uniform distribution.
        \State Update opinions:
        \State $w_i^{n+1} = w_i^{n} + \epsilon(w_j^{n} - w_i^{n}) + \sqrt{\sigma} D(w_i^{n}) \cdot \zeta_{i}^{n}$
        \State $w_j^{n+1} = w_j^{n} + \epsilon(w_i^{n} - w_j^{n}) + \sqrt{\sigma} D(w_j^{n}) \cdot \zeta_{j}^{n}$
    \EndFor

    \State \textbf{Physical interactions:} Sample $N_p = \text{round}(n_{\text{physical}} / 2)$ pairs of agents $(i,j)$ without repetition.
    \State Generate mask to locate inter-departmental interactions between a susceptible (S) and infected (I) agent.
    \ForAll{pairs $(i, j)$}
        \If{S-I interaction}
            \State Compute infection probability: $p = K(w_{\text{susc}}, w_{\text{inf}}, \alpha, \beta)$
            \State Sample from Bernoulli distribution $\mathcal{B}(p)$
            \If{sample is 1}
                \State $w_{\text{susc}}^{n+1} = w_{\text{susc}}^{n} + h(1 - w_{\text{susc}}^{n})$
                \State Update state: $s_{\text{susc}} \rightarrow s_{\text{inf}}$
            \Else
                \State $w_{\text{susc}}^{n+1} = w_{\text{susc}}^{n} + h(-1 - w_{\text{susc}}^{n})$
            \EndIf
        \EndIf
    \EndFor

    \State \textbf{Recovery:}
    \ForAll{infected agents $s_i$}
        \State Sample from Bernoulli distribution $\mathcal{B}(\gamma \cdot \epsilon)$
        \If{sample is 1}
            \State $w_{\text{inf}} = w_{\text{inf}} + h(-1 - w_{\text{inf}})$
            \State Update state: $s_{\text{inf}} \rightarrow s_{\text{susc}}$
        \EndIf
    \EndFor
\EndFor
\end{algorithmic}
\end{algorithm}

\begin{table} 
\centering
\caption{Choice of model parameters.}
\label{tab:parameters}
\begin{tabular}{@{} l c c @{}}
\toprule
\textbf{Parameter} & \textbf{A} & \textbf{B} \\
\midrule
$\sigma^{2}$ & $0.0$ & $10^{-1}$ \\
$\alpha$ & $1.0$ & $1.0$ \\
$\beta$ & $0.25$ & $0.25$ \\
$\gamma$ & $0.5$ & $0.5$ \\
$h$ & $0.1$ & $0.1$ \\
$\tau_P$ & $0.25$ & $0.25$ \\
$\rho_I(0)$ & $0.2$ & $0.8$ \\
$m_0$ & $-0.2$ & $0.0$ \\
$\nu$ & $0.0$ & $0.0$ \\
\bottomrule
\end{tabular}
\end{table}

We consider $\sigma^2 = 0,10^{-1}$ to showcase the consistency with the two macroscopic models with $\tau_S = 1, 10^{-1}, 10^{-3}$. We take $N=10^6$ agents and the choice of the model parameters is summarised in Table~\ref{tab:parameters}. 

\begin{figure}
    \centering

    \begin{subfigure}{0.48\textwidth}
        \centering
        \includegraphics[width=\linewidth]{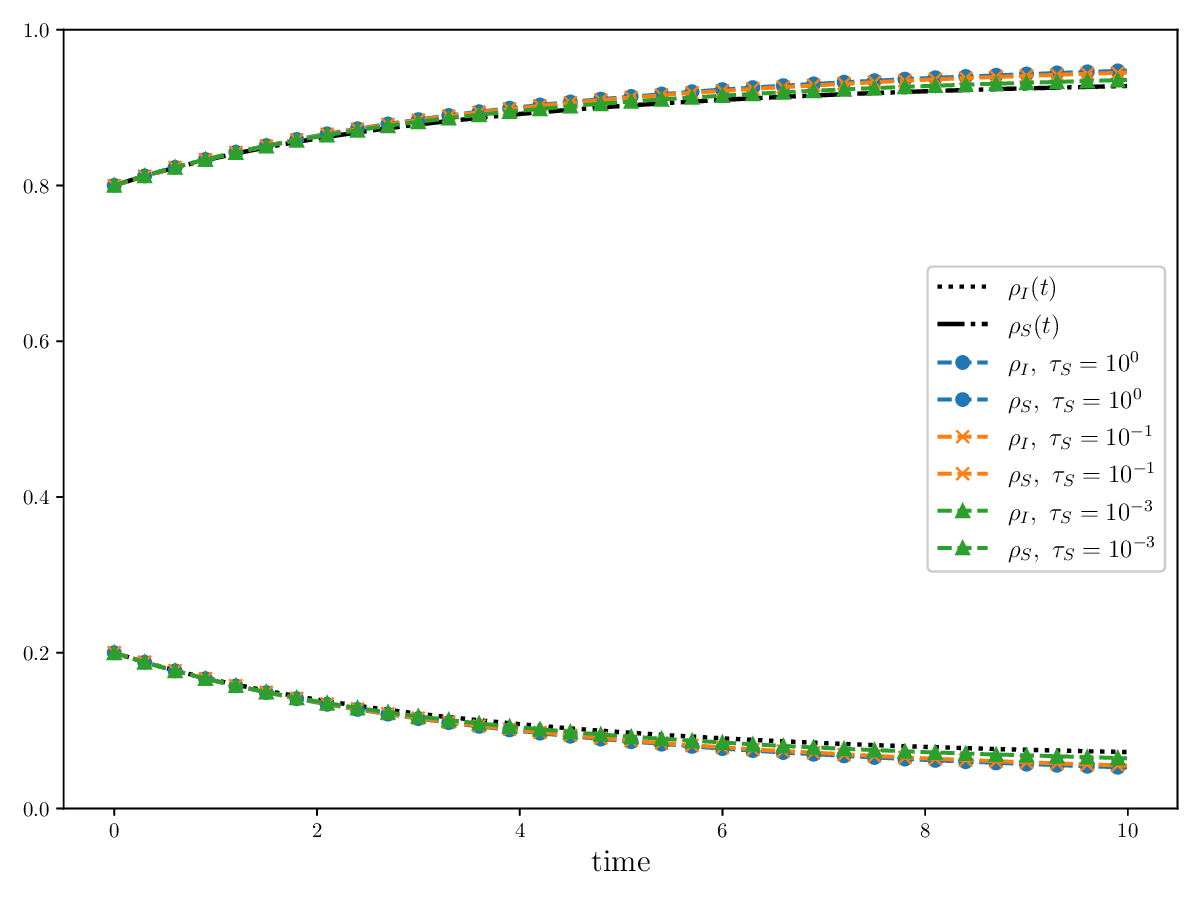}
        \caption{}
    \end{subfigure}
    \hfill
    \begin{subfigure}{0.48\textwidth}
        \centering
        \includegraphics[width=\linewidth]{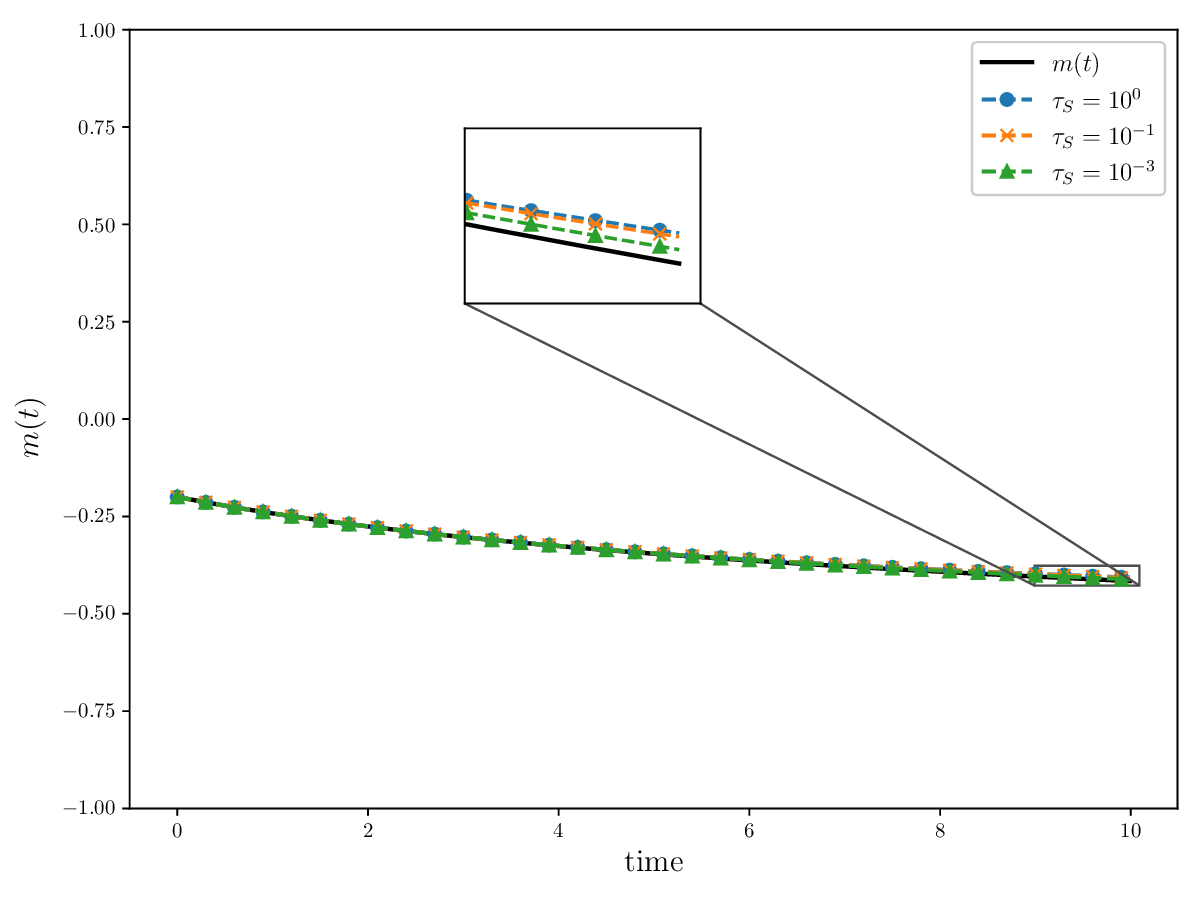}
        \caption{}
    \end{subfigure}

    \begin{subfigure}{0.48\textwidth}
        \centering
        \includegraphics[width=\linewidth]{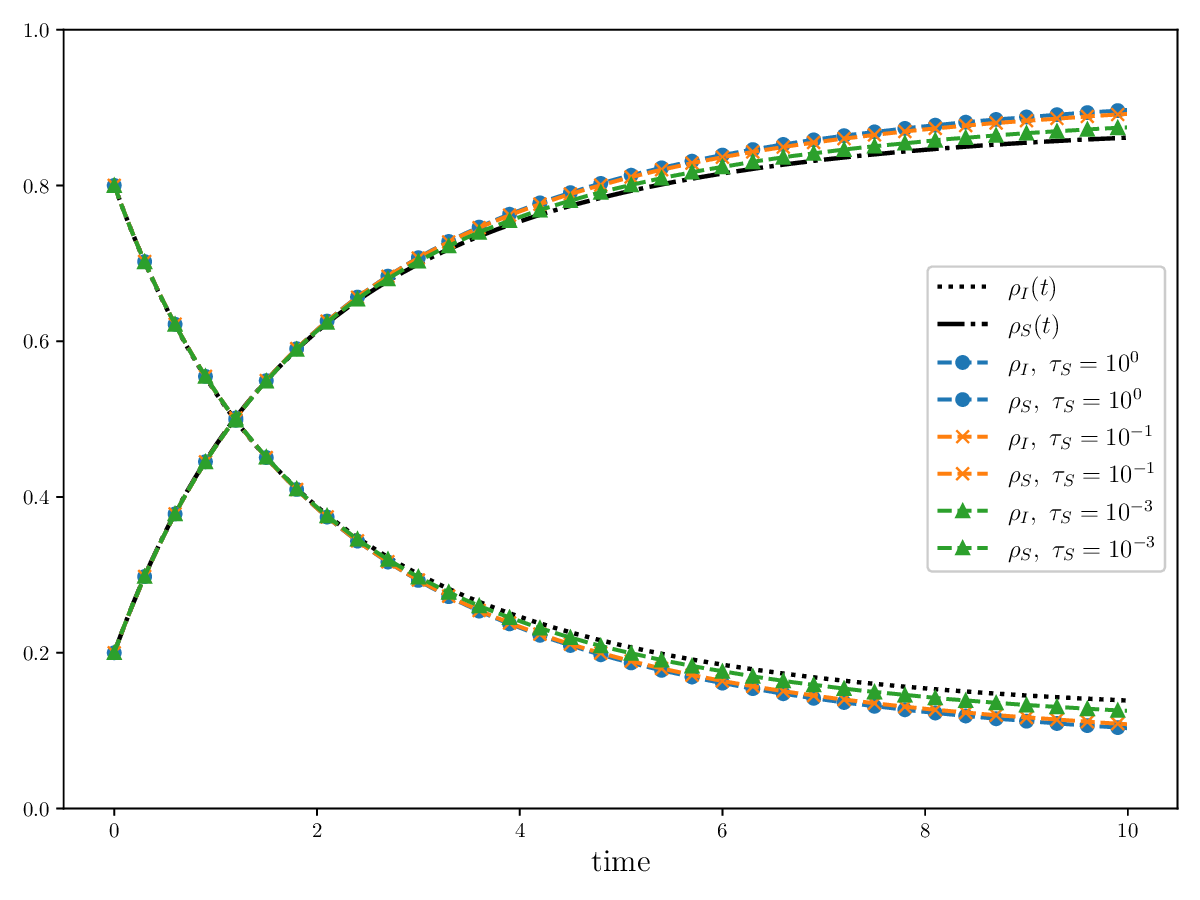}
        \caption{}
    \end{subfigure}
    \hfill
    \begin{subfigure}{0.48\textwidth}
        \centering
        \includegraphics[width=\linewidth]{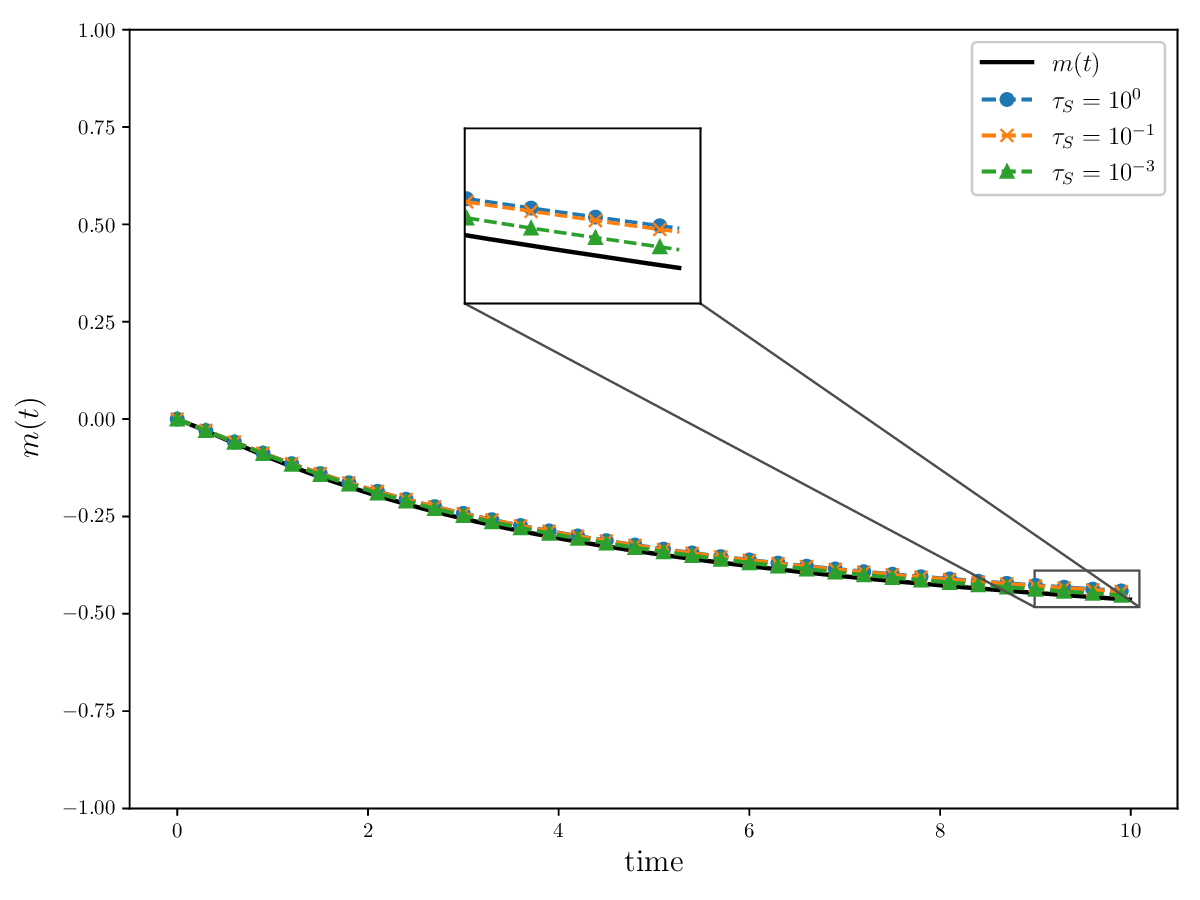}
        \caption{}
    \end{subfigure}

  \caption{Evolution of macroscopic quantites obtained by the time evolution of (\ref{Approx_Dirac})-(\ref{HighFreqGrazing2}) and the DSMC scheme shown in Algorithm~\ref{algo:DSMC}. Simulations in the first row are obtained using parameter set A (Table~\ref{tab:parameters}), while those in the second row correspond to parameter set B. We consider $N=10^6$ agents between $t =[0,5]$ with $\Delta t=10^{-3}$ and different values of the rate of social interactions $\tau_{S} = 1,10^{-1},10^{-3}$.}
  \label{fig:consistency}
\end{figure}

In Figure~\ref{fig:consistency} we present the time evolution of the infected fraction and the population mean opinion for both cases. The results indicate that, as the rate of social interactions increases, the microscopic dynamics converges to the corresponding macroscopic systems, as expected.

\subsection{Impact of reaction-type dynamics on epidemic evolution}

In this test, we study the role of the parameter $h_{r}$, which we refer to as the \textit{reaction parameter}. This parameters controls the change in opinion after a physical interaction has occured, where individuals become more or less cautious depending on the outcome of the interaction as shown in Equations~\ref{eq:micro_SS} and \ref{eq:micro_SI}. We rewrite Equation~\ref{HighFreqGrazing2} as follows

\begin{equation}
\label{modified_mean}
\begin{split}
\frac{d}{dt} m_t =& - h (\nu + \gamma \rho_t^I) + \dfrac{1}{\tau_P} h_{r} \rho^I_t \rho^S_t(1+m_t) + \\
&\frac{2}{\tau_P}\,\rho^I_t \rho^S_t h_{r} \iint \kappa(w,w_*) S[m_t](w) S[m_t](w_*) \,dw\,dw_* 
\end{split}
\end{equation}
in order to differentiate the various contributions that affect the evolution of the population’s mean opinion. 
To investigate the impact of this dynamic in particular, we neglect the change in opinion after an infected individual recovers by setting $h=0$. \
In Figure~\ref{fig:alpha_infection} we show the contagion probability in terms of the population's mean opinion considering different choices of the parameter $\alpha$. We notice that the contagion probability decreases as we intensified the impact of the opinions in the contagion probability. \

To understand the impact of the reaction parameter on the epidemic evolution we consider $h_{r} = 0.1,0.5$ and set $\alpha=1, \beta=0.1, \gamma=0.1$ with $m(0)=0.5$ and study the long time behaviour of Equation~\ref{modified_mean} as shown in Figure~\ref{fig:reaction_impact}. For this setting we observe that at the beginning the contagion of probability is low. Therefore, most physical interactions contribute to lowering the mean opinion, since after each physical interaction in which individuals do not become infected, they tend to become more careless. This mechanism explains why we observe a lowering in the population's mean opinion at the beginning. From Figure~\ref{fig:alpha_infection}, we notice that a decrease in the mean opinion leads to a higher probability of contagion. For $h_{r} = 0.1$, the mean opinion concentrates around $m = 0$ and the asymptotic number of infected individuals tends to zero. Nevertheless, by increasing the reaction parameter to $h_{r} = 0.5$, we observe larger opinion changes, which result in a higher probability of contagion that ultimately lead to an epidemic outbreak. This shows that reaction-type dynamics—where individuals become overconfident when they are not infected after a physical interaction—can result in an increased number of asymptotic infected individuals.

\begin{figure}
    \centering
    \includegraphics[width=0.6\textwidth]{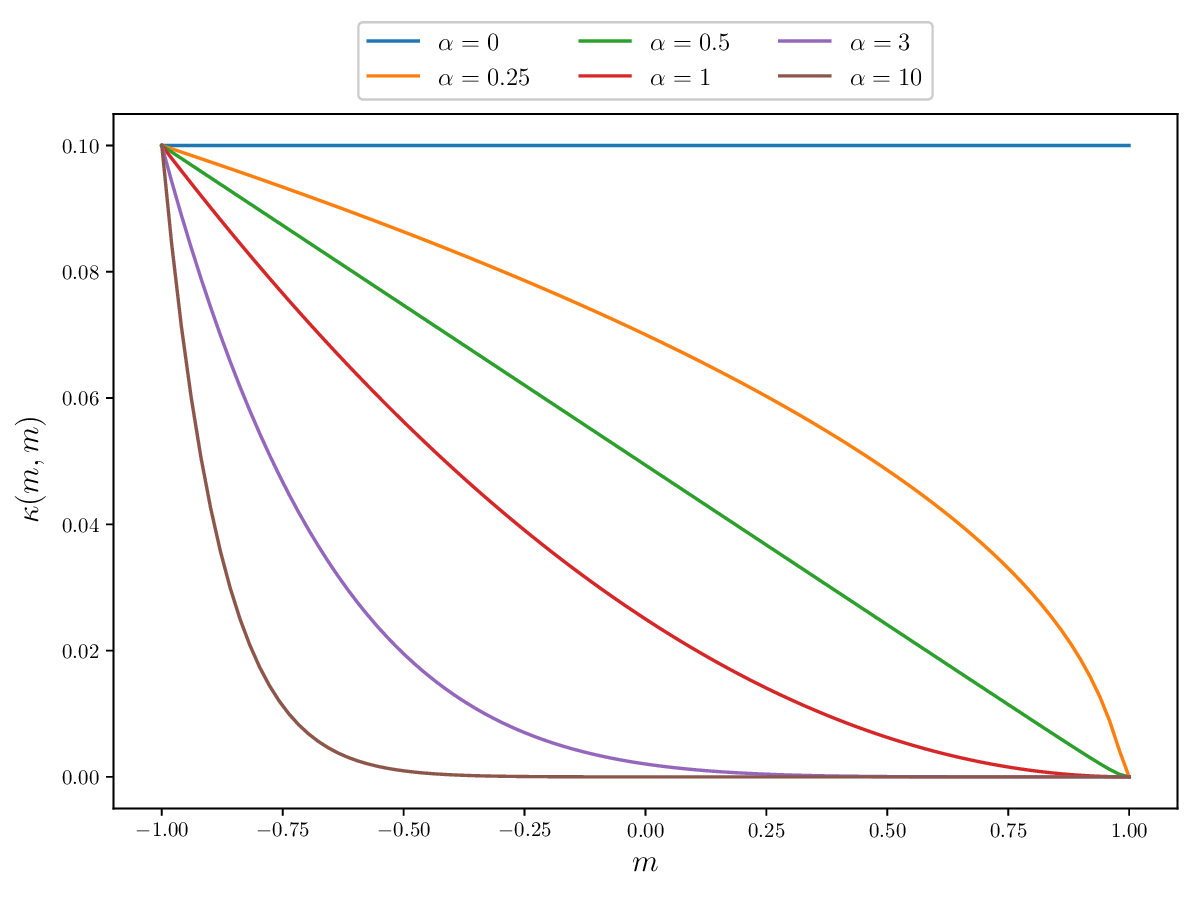}

    \caption{Contagion probability in terms of the total population mean opinion for different values of the parameter $\alpha$. Its maximum value is equal to $\beta$ which in this case is set to $0.1$.}
    \label{fig:alpha_infection}
\end{figure}

\begin{figure}
    \centering

    \begin{subfigure}{0.48\textwidth}
        \centering
        \includegraphics[width=\linewidth]{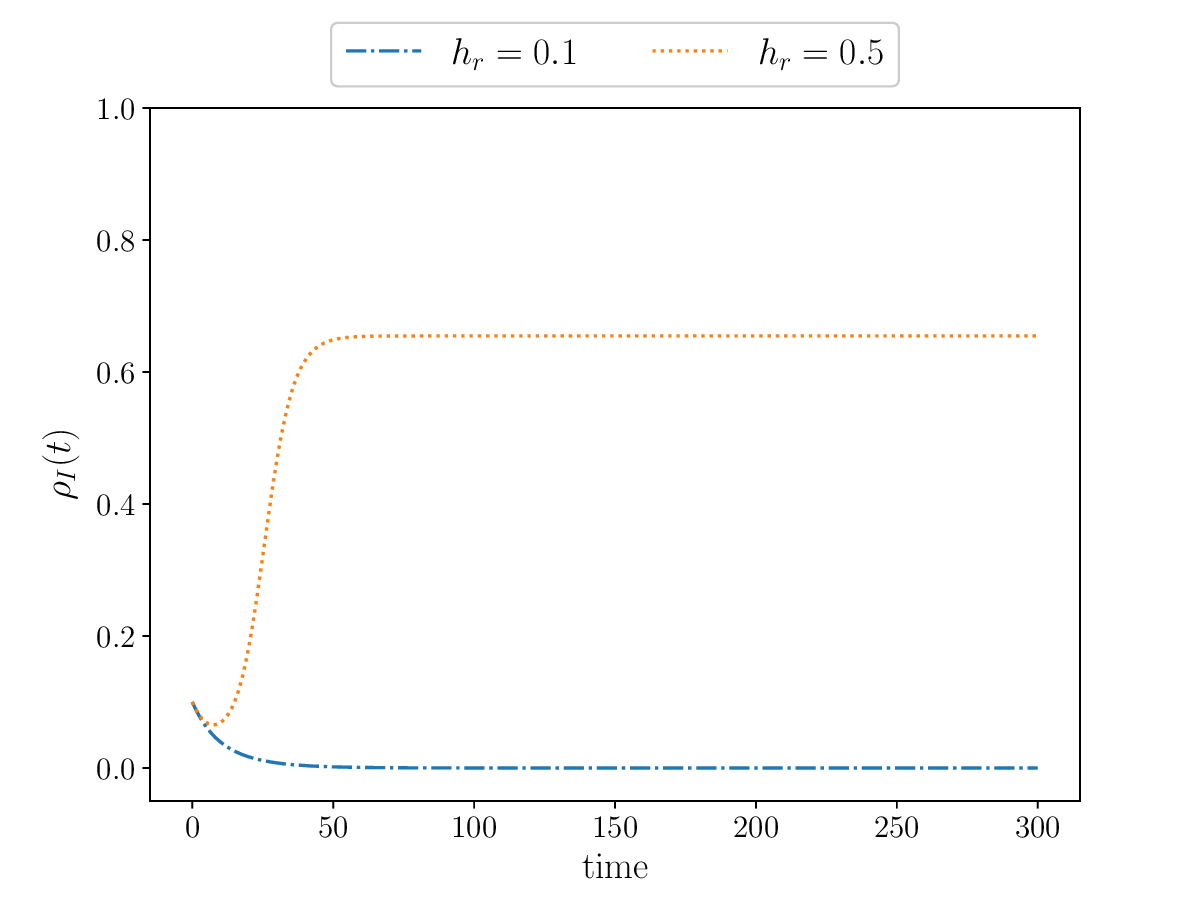}
    \end{subfigure}
    \hfill
    \begin{subfigure}{0.48\textwidth}
        \centering
        \includegraphics[width=\linewidth]{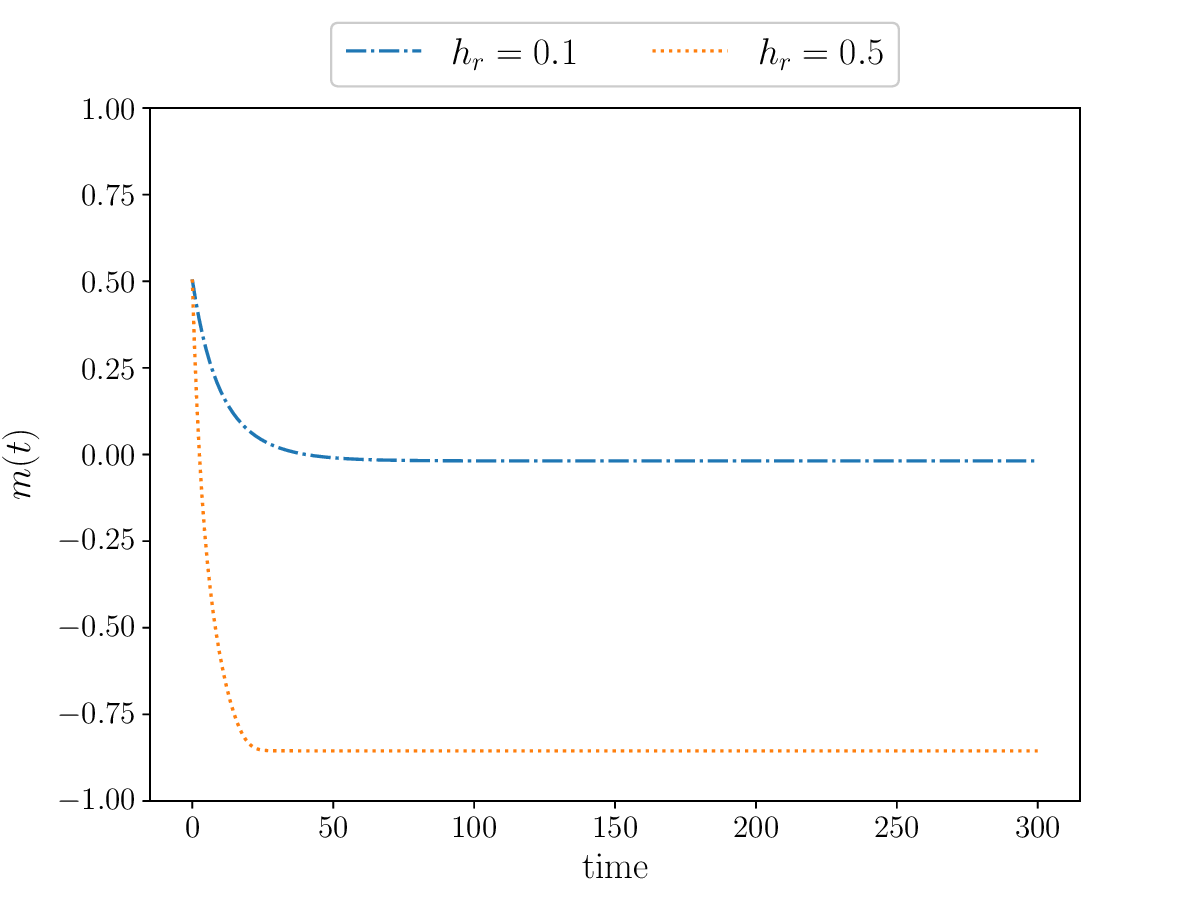}
    \end{subfigure}

  \caption{Evolution of the macroscopic quantities given by Equations~\ref{modified_mean} and \ref{HighFreqGrazing2} for different values of the reaction parameter $h_{r}$. We consider $\alpha=1, \beta=0.1, \gamma=0.1$ with $m(0)=0.5$ between $t =[0,300]$ with $\Delta t=10^{-2}$. We notice that increasing the reaction parameter leads to a higher number of asymptotic infected individuals. This is due to the fact that after each physical interaction in which individuals do not become infected, they tend to become more careless, which results in a higher probability of contagion leading to a higher number of asymptotic infected individuals.}
  \label{fig:reaction_impact}
\end{figure}


\section{Conclusions}

In this work we have introduced and analysed a kinetic multi-agent framework describing the mutual interplay between epidemic spreading and opinion dynamics. The model captures a bidirectional feedback mechanism in which individual compliance influences transmission events, while epidemic encounters asymmetrically reshape social attitudes. In particular, infection triggers protective-like behaviour, whereas unsuccessful transmission events foster opinion radicalisation. This simple but non-symmetric updating rule generates nontrivial collective effects emerging from the coupling of behavioural and epidemiological variables.

Within a prototypical SIS setting, we derived a macroscopic kinetic description of the interacting system and investigated a fast social-interaction regime leading to a reduced system of differential equations. The rigorous quantification of convergence toward the reduced dynamics through a modified Wasserstein distance provides a solid analytical foundation for the validity of the macroscopic approximation. In the noiseless tendency-to-compromise regime, this reduction can be interpreted as a fast–slow dynamical system, where high-frequency social exchanges rapidly drive the opinion distributions toward local equilibria. Numerical simulations performed through a particle-based time-splitting Monte Carlo scheme confirm the consistency of the reduced ODE approximation and its ability to capture the long-time qualitative behaviour of the full kinetic system.

From a methodological perspective, the proposed strategy suggests a systematic pathway for deriving low-dimensional ODE models from high-dimensional kinetic descriptions of coupled social–epidemic systems. Extending this approach to more articulated compartmental structures, such as SEIR or SIRS models, would contribute to the development of a general reduction framework. Although the current convergence proof is restricted to a two-equation system, its structural arguments are expected to remain adaptable to higher-dimensional settings. A major open challenge concerns the relaxation of the noiseless assumption in the opinion dynamics. Incorporating stochastic fluctuations would lead to local equilibria described by Beta-type distributions rather than Dirac masses, providing a more realistic representation of uncertainty and heterogeneity in social behaviour.

Beyond its analytical relevance, the reduced system offers a transparent and computationally efficient proxy for sensitivity analysis and parameter exploration. This aspect is particularly meaningful from a public health governance perspective, where the objective is to steer individual opinions toward protective behaviours at minimal socio-economic cost while ensuring epidemic extinction. The reduced dynamics allow one to identify the most influential drivers of the coupled system in a tractable manner. However, the alignment between optimal control strategies derived from the original kinetic formulation and those obtained from the reduced ODE model deserves further investigation, in order to rigorously assess the reliability of the approximation in decision-making contexts.

Future research will therefore address both the qualitative analysis of the reduced dynamics and the extension to stochastic and more structured interaction mechanisms, with the ultimate goal of evaluating the robustness of the epidemic–opinion feedback when confronted with more realistic modelling assumptions and empirical data.

\section*{Data Availability Statement}
All data connected with this study are available upon request. 

\section*{Acknowledgements}
The research of this paper has been undertaken within the activities of the GNFM group of INdAM (National Institute of High Mathematics). H.T. and M.Z. acknowledge partial support from the PRIN2022PNRR project
No.P2022Z7ZAJ, European Union-NextGenerationEU. M.Z. acknowledges partial support by ICSC - Centro Nazionale di Ricerca in High Performance Computing, Big
Data and Quantum Computing, funded by European Union-NextGenerationEU. We acknowledge support provided by CUIA (Consorzio Interuniversitario Italiano per l'Argentina) for funding mobility and scientific exchange between the participating institutions.

\appendix

\section{Well-posedness of the kinetic system: Proof of Theorem \ref{Thm:well-posedness}.}
\label{Appendix:well-posedness}


Denote $\M([-1,1])$ the space of finite Borel measures on $[-1,1]$, $\M_+([-1,1])$ the subset of nonnegative measures, 
and $\P([-1,1])$ the subset of probability measures. They are endowed with the total variation (TV) norm. 

We now rewrite the system \eqref{full.system.weak} as just one equation for the probability measure 
$g_t:=f_t^I\otimes \delta_I + f_t^S\otimes \delta_S$ on $K:=[-1,1]\times \{I,S\}$. 
Each of the mechanisms considered here, namely memory fading, recovery, contagion-based and pure social interaction, 
can be described by one kernel for $g_t$. 
Notice test-functions are now continuous functions $\phi$ on $[-1,1]\times \{I,S\}$ i.e. $\phi$ is of the form 
$\phi(w,c)=\phi_S(w) \mathds{1}_{c=S}+\phi_I(w) \mathds{1}_{c=I}$ with $\phi_S(\cdot),\phi_I(\cdot)$ test functions. 
Then,  we get
\[
\int_{[-1,1]}\phi\,dg_t = \int_{[-1,1]} \phi_S \,df_t^S + \int_{[-1,1]} \phi_I \,df_t^I.
\]
We can then define the TV norm of $g_t$.  The memory fading operators $Q_{S\to S}$ and $Q_{I\to I}$ can be equivalently described by $Q_{H\to H}$, $H \in \mathcal C$, defined by 
\[
 (Q_{H\to H},\phi) = \nu \int_{[-1,1]}\langle\phi(w'_H,c)-\phi(w,c)\rangle\,dg_t(w,c). 
 \]
The recovery kernels $Q_{I\to S}^L$ and $Q_{I\to S}^G$ yields the kernel $Q_{I\to S}$ defined by 
\[
 (Q_{I\to S},\phi) = \gamma\int_{[-1,1]}(\phi(w'_{I\to S},c') - \phi(w,c))\mathds{1}_{c=I,c'=S} \,dg_t(w,c). 
 \]
The pure social interactions kernels $Q_{II}, Q_{IS},Q_{SI},Q_{SS}$ correspond to $Q_{CC}$ defined by 
\[
(Q_{CC_*},\phi) = \frac{1}{\varepsilon}\iint_{[-1,1]^2}(\phi(w_{cc_*},c)-\phi(w,c)) \, dg_t(w,c)dg_t(w_*,c_*). 
\]
Eventually the contagion-based social interaction kernels $Q^G_{S+I\rightarrow I+I}$, $- Q^L_{S+I\rightarrow I+I}$ and 
$Q_{S+I\rightarrow S+I}$ correspond to $Q^G_{S+I\rightarrow S/I+I}$ defined as 
\begin{eqnarray*} 
 &&(Q_{S+I\rightarrow S/I+I},\phi) = 
\frac1\tau \iint_{[-1,1]^2} \Big\{(\phi(w'_{S\to I},c')-\phi(w,c))\kappa(w,w_*) \mathds{1}_{c'=I}  \\ 
&&                + (\phi(w'_{S\to S},c')-\phi(w,c))(1-\kappa(w,w_*)) \mathds{1}_{c'=S} \Big\}
               \mathds{1}_{c=S,c_*=I}\,dg_t(w,c)dg_t(w_*,c_*)
\end{eqnarray*} 
Thus the system \eqref{full.system.weak} can be rewritten as a unique equation for $g_t$ given by 
\begin{equation}\label{Equ100}
 \p_t g_t = Q(g_t):= Q_{C\to C} + Q_{I\to S} + Q_{CC} + Q_{S+I\rightarrow S/I+I} 
 \end{equation}
i.e. 
\begin{equation}\label{FullSYstemOne}
\begin{split}
& \frac{d}{dt} \int_{[-1,1]}\phi(w,c) \,dg_t(w,c) = \int_K \phi\, Q(g_t)
:= \nu \int_{[-1,1]} \langle\phi(w'_c,c)-\phi(w,c)\rangle\,dg_t(w,c) \\ 
& + \gamma\iint_{[-1,1]^2} (\phi(w'_{I\to S},c') - \phi(w,c))\mathds{1}_{c=I,c'=S} \,dg_t(w,c) \\
& +\frac{1}{\varepsilon} \iint_{[-1,1]^2}(\phi(w_{cc_*},c)-\phi(w,c)) \, dg_t(w,c)dg_t(w_*,c_*) \\
& + \frac1\tau \iint_{[-1,1]^2} \Big\{(\phi(w'_{S\to I},c')-\phi(w,c))\kappa(w,w_*) \mathds{1}_{c'=I}  \\ 
&                + (\phi(w'_{S\to S},c')-\phi(w,c))(1-\kappa(w,w_*)) \mathds{1}_{c'=S} \Big\}
               \mathds{1}_{c=S,c_*=I}\,dg_t(w,c)dg_t(w_*,c_*)
\end{split} 
\end{equation}
for any $\phi\in C([-1,1]\times \{S,I\})$.

Given an initial condition $g_0\in \P(K)$, we thus look for $g\in C([0,+\infty),\P(K))$ satisfying   
$$ g_t = g_0 + \int_0^t Q(g_s)\,ds, \qquad t>0, $$
and $g_{t=0}=g_0$. 
We briefly sketch the proof. 

The operator $Q:\M(K)\to \M(K)$ defined in \eqref{Equ100} satisfies the following properties: for any $f,g\in \M(K)$,  
\begin{enumerate} 
\item[(i)] $Q(f)$ has total mass 0, 
\item[(ii)] there exists $L>0$ such that if $\|f\|_{TV},\|g\|_{TV}\le R$, then 
$$ \|Q(f)-Q(g)\|_{TV} \le LR\|f-g\|_{TV}. $$ 
\item[(iii)] We can write $Q=Q^+-Q^-$ with $Q^\pm(f)\ge 0$ if $f\ge 0$. 
\end{enumerate}
Property (i) follows taking $\phi=1$ so that $\int 1.Q(f)=0$. 
For (ii), it is easily seen that for any $\phi\in C(K)$, $\|\phi\|_\infty\le 1$, we have 
$$ \int \phi\,(Q(f)-Q(g)) \le LR\|f-g\|_{TV}, $$
with $L:=2(\nu+\gamma+\frac{1}{\varepsilon} +  \frac{3}{\tau})$. 
The result follows taking the sup over $\phi$. 
It follows in particular from (ii) that if $f\in C([0,T],\M(K))$ for some $T>0$, then $Q(f)$ defined by $Q(f)_t:=Q(f_t)$ is also continuous in $t$. 
Property (iii) is clear since $\kappa(w,w_*)\in [0,1]$. 
Notice that for $f\in \P(K)$, $Q^-(f) \le L'f$ with 
$L'=\nu + \gamma + \frac{1}{\varepsilon} + \frac{3}{\tau}$. 
It follows in particular that for $h\in [0,1/L']$, we have for any $f\ge 0$ that $f+hQ(f) \in \P(K)$. 
Indeed $f+hQ(f)$ hat total mass 1 and is nonnegative since  $f+hQ(f) \ge f-hQ^-(f) =(1-hL')f\ge 0$.

The above properties of $Q$ allows to apply Bressan' techniques \cite{Bressan} and obtain the existence of a solution $g$, necessarily unique.
This solution is obtained as the limit of a sequence of approximate solutions built in the same spirit as the Euler scheme in numerical analysis. 
We refer to \cite{PedrazaPinascoSaintier}[Section 4.1] for more details.

\section{Convergence to consensus: Proof of Prop. \ref{prop:consensus}. }
\label{proof:prop_consensus}

Consider the pure opinion dynamic \eqref{Def_wCC} in absence of self-thinking: 
\begin{equation}\label{SystemPureSocial10}
\begin{split}
\frac{d}{dt}\int \phi\,df_t^I  & = \frac{1}{\eps} \iint [\phi(w+\xi P(w,w_*)(w_*-w))-\phi(w)]\,df_t^I(w)df_t(w_*), \\
\frac{d}{dt}\int \phi\,df_t^S  & = \frac{1}{\eps} \iint [\phi(w+\xi P(w,w_*)(w_*-w))-\phi(w)]\,df_t^S(w)df_t(w_*).  
\end{split}
\end{equation} 
Summing these two equation gives 
$$ \frac{d}{dt}\int \phi\,df_t = \frac{1}{\eps} \iint [\phi(w+\xi P(w,w_*) (w_*-w))-\phi(w)]\,df_t(w)df_t(w_*) 
\qquad \phi\in C([-1,1]). $$ 
Since $P$ is symmetric, i.e. $P(w,w_*)=P(w_*,w)$ for any $w,w_*\in [-1,1]$, 
taking $\phi(w)=w$ shows that 
$\langle w\rangle := \int w\,df_t$  is constant in time. 
Next, taking $\phi(w)=w^2$, we obtain the evolution of the variance of $f_t$  as follows
\begin{eqnarray*}
 \frac{\eps}{2\gamma} \frac{d}{dt} Var[f_t] 
 = \frac{\eps}{2\gamma}  \frac{d}{dt} \int w^2\,df_t(w)
& = & \iint ww_* P(w,w_*)(1-\xi P(w,w_*)) \,df_t(w) df_t(w_*) \\
&& - \iint w^2 P(w,w_*)(1-\xi P(w,w_*)) \,df_t(w) df_t(w_*) 
\end{eqnarray*} 
where we used the symmetry of $P$. Using again the symmetry of $P$, we can change $w$ for $w_*$ in the 2nd integral. Thus 
\begin{eqnarray*}
 \frac{\eps}{\gamma} \frac{d}{dt} Var[f_t] 
 & = & - \iint (w^2 + w_*^2 - 2ww_*) P(w,w_*)(1-\xi P(w,w_*)) \,df_t(w) df_t(w_*) \\
& = & - \iint (w - w_*)^2 P(w,w_*)(1-\xi P(w,w_*)) \,df_t(w) df_t(w_*). 
\end{eqnarray*} 
Supposing that  $P(w,w_*)>0$ if $w\neq w_*$, and $\xi<1/\|P\|_\infty$, we deduce
$$  \frac{\eps}{\xi} \frac{d}{dt} Var[f_t]  
\le - C\iint (w - w_*)^2 P(w,w_*) \,df_t(w) df_t(w_*). 
$$
Notice the integrand is zero only if $w=w_*$. 
Thus the variance of $f_t$ strictly decreases until $w=w_*$ $f_t(dw)\otimes f_t(dw_*)$-a.e. i.e. until  $f_t$ is a Dirac mass. 
It follows that $f_t\to \delta_{\langle w\rangle}$. 
If morover $P\ge C>0$ for some positive constant $C$, then this convergence is exponentially fast since 
the previous estimate becomes 
$$   \frac{d}{dt} Var[f_t]  
\le - \frac{C}{\eps}\iint (w - w_*)^2 \,df_t(w) df_t(w_*) 
= - \frac{2C}{\eps} Var[f_t]
$$
and we deduce that $Var[f_t]\le Var[f_0] e^{-\frac{C' t}{\eps}}$. 

Eventually, recall that the probability measures on $[-1,1]$ form a  compact set for the weak convergence. 
Any limit $(\mu,\nu)$ of $(f_t^I, f_t^S)$ along a subsequence thus satisfies $\mu+\nu=\delta_{\langle w\rangle}$ 
and is therefore of the form $(\mu,\nu)=(m\delta_{\langle w\rangle}, (1-m)\delta_{\langle w\rangle})$ 
for some $m\in [0,1]$. 
Thus $(f_t^I, f_t^S)$ converges to the set $\{(m\delta_{\langle w\rangle},(1-m)\delta_{\langle w\rangle}),\, m\in [0,1]\}$. Since $\rho^I$ and $\rho^S$ are constant as can be seen taking $\phi\equiv 1 $ in 
\eqref{SystemPureSocial10}, 
we deduce that $f_t^I\to \rho^I \delta_{\langle w\rangle}$ and 
$f_t^S\to \rho^S\delta_{\langle w\rangle}$.

\section{Dimension reduction: Proof of Theorem \ref{Thm_Approx}. } \label{Appendix_Section_DimReduction}

In this Appendix we give a proof of Theorem \ref{Thm_Approx} concerning the approximation of a system like 
\begin{equation}\label{FullSystem10}
\begin{aligned}
\partial_t f_t = L[f_t,g_t] + \frac1\varepsilon Q[f_t,h_t],\\ 
\partial_t g_t = M[f_t,g_t] + \frac1\varepsilon Q[g_t,h_t],
\end{aligned}
\end{equation} 
in the limit $\varepsilon\to 0$ by a system of ODEs for $\rho_t=\int\,df_t$ and the mean opinion $m_t$ of $h_t$. 
We divide the proof into two main Steps. We first consider a single equation like 
$$ \partial_t f_t^\varepsilon = Lf_t^\varepsilon + \frac{1}{\varepsilon}Q[f_t^\varepsilon], $$ 
and then tackle the system \eqref{FullSystem10}. 
In both cases, we will suppose that given an initial condition, there is a
unique solution for each $\varepsilon$, and focus on its asymptotic behavior as $\varepsilon\to 0$.

\subsection{Dimension reduction for a single equation.}

We first consider an equation of the form 
\begin{equation}\label{Appendix_SingleEqu}
\partial_t f_t^\varepsilon = Lf_t^\varepsilon + \frac{1}{\varepsilon}Q[f_t^\varepsilon]
\end{equation}
with an initial condition $f_0\in \P([-1,1])$. 
Here $L,Q:\P([-1,1])\to \M([-1,1])$. 
Let us assume that there exists a unique solution $f^\varepsilon\in C([0,+\infty),\P([-1,1])$. 

For ease of notation, we will frequently denote $(f,\phi)$ the integral $\int \phi(x)\,df(x)$ of a test-function $\phi$ against a measure $f$. 

To study the behavior of $f^\varepsilon$ as $\varepsilon\to 0$, we assume that $L$ satisfies 
\begin{enumerate} 
\item[(H1)] For and $f\in \P([-1,1])$, $Lf$ has mass 0,
\item[(H2)] The mean value of $Lf$ is Lipcshitz in $f$: there exists $C_L>0$ such that 
$$|(Lf-Lg,x)|\le C_L W_2(f,g) \qquad \text{for any $f,g\in \P([-1,1])$.} $$
\item[(H2)] The total variation of $Lf$ is bounded uniformly in $f\in\P([-1,1])$: there exists $C_{TV}>0$ such that 
$$ \|Lf\|_{TV} \le C_{TV} \qquad\text{for any $f\in \P([-1,1])$.} $$
\end{enumerate}
Concerning $Q$ we suppose that 
\begin{enumerate} 
\item[(H3)] For any $f\in \P([-1,1])$, the measure $Q[f]$ has mass and mean zero: 
$$(Q[f],1) = (Q[f],x)=0. $$ 
\item[(H4)] $Q$ is contractive: there exists $\lambda >0$ such that for any $f\in \P([-1,1])$, 
$$ (Q[f],x^2)  \le -\lambda Var[f],$$
where the variance of $f$ is $Var[f]=(f,x^2)-(f,x)^2$. 
\end{enumerate}
An example of kernel $Q$ is the tendency-to-compromise operator
$$ (Q[f],\phi) = \iint [\phi(x+\gamma(x_*-x))-\phi(x)]\,df(x)df(x_*).$$
In that case (H3) clearly holds and (H4) is verified with $\lambda = \gamma(1-\gamma)$, $\gamma\in (0,1)$.

Due to the contractivity  property, we expect that when $\varepsilon\to 0$, 
the solution $f_t^\varepsilon$ will rapidly approach (on the time time scale $t/\varepsilon$) 
the slow manifold $\{\delta_m,\,m\in\mathbb{R}\}$ and then stay close to it. 
An heuristic reasoning as in section \ref{Heursitic_Dirac} shows that $f_t^\varepsilon\simeq \delta_{m_t}$ 
where $a$ is the solution of the ODE
\begin{equation}\label{ODE_a}
\begin{split}
\frac{d}{dt}m_t & = (L\delta_{m_t},x), \\  
m_0& = \int x\,df_0(x)
\end{split}
\end{equation}
The following result justifies this intuition and provide an estimate of the distance between  $f_t^\varepsilon$ and $\delta_{m_t}$
in terms of the Monge-Kantorovich (or Wasserstein) distance $W_2$. The $W_2$ distance between two probability measures $f,g\in \P([-1,1])$ 
is defined as 
$$ W_2(f,g)^2=\inf_\pi \iint_{[-1,1]\times [-1,1]} |x-y|^2 \,d\pi(x,y), $$
where the infimum is taken over all the probability measures $\pi\in \P([-1,1]\times [-1,1])$ with marginals $f$ and $g$. 
We refer to the book \cite{Villani03} for more details on these distances. 

\begin{thm}\label{Appendix_Thm_1Eq}
Given an initial condition $f_0\in \P([-1,1])$, suppose that \eqref{Appendix_SingleEqu} has a unique solution $f^\varepsilon$. 
Denote $a$ the solution of \eqref{ODE_a}. 
If assumption (H1) through (H4) hold then for any $t\ge 0$, 
$$ W_2(f_t^\varepsilon, \delta_{m_t})
\le C(e^{-\lambda t/(2\varepsilon)} + \sqrt{\varepsilon}e^{Lip(L)t}), $$
where the constant $C$ depends only $Var[f_0]$, $1/\lambda$, $Lip(L)$ and $C_{TV}$.
\end{thm}

\noindent Notice in particular for any $\gamma\in (0,1)$ and $\beta\in (0,1/(2Lip(L))$, 
$$ \lim_{\varepsilon\to 0} \sup_{\varepsilon^\gamma\le t\le -\beta\ln(\varepsilon)} W_2(f_t^\varepsilon, \delta_{m_t}) = 0. $$

\begin{proof}
The mean $m^\varepsilon(t)$ of $f_t^\varepsilon$ satisfies 
$$ \frac{d}{dt} m^\varepsilon(t) =  \frac{d}{dt} (f_t^\varepsilon,x)= (Lf_t^\varepsilon,x), $$
where we used (H3). 
Thus 
\begin{eqnarray*}  
\frac{d}{dt} Var[f_t^\varepsilon] 
& = & \frac{d}{dt} (f_t^\varepsilon,x^2) - 2(f_t^\varepsilon,x) \frac{d}{dt} (f_t^\varepsilon,x) \\
& = & (Lf_t^\varepsilon,x^2) + \frac1\varepsilon (Q[f_t^\varepsilon],x^2) -2 (f_t^\varepsilon,x)(Lf_t^\varepsilon,x).  
\end{eqnarray*}
Notice that 
$$ |(Lf_t^\varepsilon,x^2) -2 (f_t^\varepsilon,x)(Lf_t^\varepsilon,x)| 
\le 3\|Lf_t^\varepsilon\|_{TV} \le 3 C_{TV}$$ 
by (H2). 
We then obtain with (H4) that 
$$ \frac{d}{dt} Var[f_t^\varepsilon] \le - \frac\lambda\varepsilon Var[f_t^\varepsilon] + 3 C_{TV} $$ 
i.e.
$$ W_2(f_t^\varepsilon,\delta_{m^\varepsilon(t)})^2 
= Var[f_t^\varepsilon] 
\le e^{-\lambda t/\varepsilon} Var[f_0] + \frac{3\varepsilon  C_{TV}}{\lambda}.  
$$
Thus 
\begin{equation}\label{Appendix_Eq10}
 W_2(f_t^\varepsilon,\delta_{m^\varepsilon(t)}) 
\le C(e^{-\lambda t/(2\varepsilon)}  + \sqrt{\varepsilon})). 
\end{equation}

Since 
$$ W_2(f_t^\varepsilon,\delta_{m_t})  \le 
W_2(f_t^\varepsilon,\delta_{m^\varepsilon(t)})  + W_2(\delta_{m^\varepsilon(t)}, \delta_{m_t}) $$
with  $W_2(\delta_{m^\varepsilon(t)}, \delta_{m_t}) = |m^\varepsilon(t)-m_t|$, we just have to prove an estimate 
like 
\begin{equation}\label{ToProve}
 |m^\varepsilon(t)-m_t| \le C\sqrt{\varepsilon}e^{Lip(L)t}
\end{equation}
to conclude the proof. 
We have 
$$ \frac{d}{dt} m^\varepsilon(t) = (Lf_t^\varepsilon,x) = (L\delta_{m^\varepsilon(t)},x) + R_\varepsilon(t) $$
with 
\begin{eqnarray*}
 |R_\varepsilon(t)| 
 &=& |(L\delta_{m^\varepsilon(t)}-Lf_t^\varepsilon,x)|
\le C_L W_2(\delta_{m^\varepsilon(t)}, f_t^\varepsilon) 
\le C(e^{-\lambda t/(2\varepsilon)}  + \sqrt{\varepsilon}))
\end{eqnarray*}
by assumption (H2) and \eqref{Appendix_Eq10}. 
Since $m^\varepsilon(0)=a(0)$ and 
$$ |(L\delta_{m^\varepsilon(s)}-L\delta_{a(s)},x)| 
  \le C_L W_2(\delta_{m^\varepsilon(s)},\delta_{a(s)}) 
  = C_L |m^\varepsilon(s) - a(s)|,
$$
we have 
\begin{eqnarray*}
|m^\varepsilon(t)-m_t|
& \le & \int_0^t|R_\varepsilon(s)|\,ds + \int_0^t |(L\delta_{m^\varepsilon(s)}-L\delta_{a(s)},x)|\,ds \\
& \le & C(\varepsilon + \sqrt{\varepsilon}t) +  C_L \int_0^t |m^\varepsilon(s)-a(s)|\,ds.
\end{eqnarray*}
We deduce \eqref{ToProve} using Gronwall inequality. 
\end{proof}

\subsection{Dimension reduction for a system}

We now consider a system like 
\begin{eqnarray*}\label{Appendix_case2Equ}
\begin{aligned}
    \partial_t f_t & = L[f_t,g_t] + \frac1\varepsilon Q[f_t,h_t],\\ 
    \partial_t g_t & = M[f_t,g_t] + \frac1\varepsilon Q[g_t,h_t],
\end{aligned}
\end{eqnarray*} 
where $f_t,g_t$ are non-negative measures such that $h_t:=f_t+g_t$ is a probability measure. 
Here $L,M:\mathcal{A}\to \M([-1,1])$ where 
$\mathcal{A}:=\{(f,g)\in \M_+([-1,1])\times \M_+([-1,1]) \text{ s.t. } f+g\in \P([-1,1])\}$, 
and $Q:\M_+([-1,1])\times \P([-1,1])\to \M([-1,1])$. 

We fix an initial condition $(f_0,g_0)\in \mathcal{A}$ and suppose there is a unique solution 
$(f^\varepsilon,g^\varepsilon)\in C([0,+\infty),\mathcal{A})$ for each $\varepsilon>0$ small. 
Denote $h_t^\varepsilon = f_t^\varepsilon + g_t^\varepsilon$. 
Our purpose here is to provide conditions on $L,M,Q$ to ensure that  $(f_t^\varepsilon,g_t^\varepsilon)$ 
is close to a pair of Dirac mass like $(r_t\delta_{m_t}, (1-r_t)\delta_{m_t})$ in the limit $\varepsilon\to 0$. 

Let us suppose that 
\begin{enumerate} 
\item[(H5)] $Q$ is linear in the 1st argument. 
\end{enumerate} 
Then adding both equations gives 
$$ \partial_t h_t^\varepsilon = L[f_t^\varepsilon,g_t^\varepsilon] + M[f_t^\varepsilon,g_t^\varepsilon] 
+ \frac1\varepsilon Q[h_t^\varepsilon,h_t^\varepsilon]. $$
Assume also $Q$ satisfies (H3) and the following analog of (H4):
\begin{enumerate} 
\item[(H4')] there exists $\lambda,\gamma >0$ and $C>0$ such that for any $(f,g)\in \mathcal{A}$, 
\begin{eqnarray*} 
&& (Q[\tilde f,h],x^2) + (Q[\tilde g,h],x^2) - 2(\tilde f,x)(Q[\tilde f,h],x) - 2(\tilde g,x)(Q[\tilde g,h],x) \\
&& \le -\lambda (Var[\tilde f]+Var[\tilde g]) + C(\tilde f-\tilde g,x)^2,
\end{eqnarray*} 
and 
$$(\tilde f-\tilde g, x)((Q[\tilde f,\tilde h],x)-(Q[\tilde g,\tilde h],x))  
\le - \gamma (\tilde f-\tilde g, x)^2, $$ 
where $h=f+g$, $\tilde f:=f/\int f$ and $\tilde g:=f/ \int g$. 
\end{enumerate} 
We will verify below that the tendency-to-compromise kernel  satisfies (H4'). 
Notice that if $(f,g)\in\mathcal{A}$ with $f=g$ then the 2nd statement in (H4') is void 
and the 1st one reduces to (H4) with the same constant $\lambda$. 
Suppose also $L$ and $M$ satisfy
\begin{enumerate}
\item[(H2')] The total variation of $L[f,g]$ and $M[f,g]$ are bounded uniformly in $(f,g)\in\mathcal{A}$: 
there exists $C_{TV}>0$ such that 
$$ \|L[f,g]\|_{TV} + \|M[f,g]\|_{TV} \le C_{TV} $$
for any $(f,g)\in\mathcal{A}$. 
\end{enumerate}
Then mimicking the first part of the proof of Theorem \ref{Appendix_Thm_1Eq} gives that 
$h_t^\varepsilon$ is close to $\delta_{m^\varepsilon(t)}$, being  $m^\varepsilon(t)$ the mean value of $h_t^\varepsilon$, in the sense that 
\begin{equation}\label{Appendix_Estim_h}
 W_2(h_t^\varepsilon,\delta_{m^\varepsilon(t)}) 
\le C(e^{-\lambda t/(2\varepsilon)}  + \sqrt{\varepsilon}). 
\end{equation}
Since $h_t^\varepsilon = f_t^\varepsilon + g_t^\varepsilon$, 
we thus expect intuitively that $f_t^\varepsilon\simeq \rho_t^\varepsilon \delta_{a^\varepsilon(t)}$ 
and $g_t^\varepsilon\simeq (1-\rho_t^\varepsilon) \delta_{a^\varepsilon(t)}$ with $\rho^\varepsilon_t$ the total mass of $f^\varepsilon_t$. 
The heuristic reasoning done in section \ref{Heursitic_Dirac} shows that $(\rho_t^\varepsilon,a^\varepsilon_t)$ should be close to 
$(\rho_t,m_t)$ defined as the solution of the limit system
\begin{eqnarray}\label{LimitEq2}
\begin{aligned}
    &\frac{d}{dt} \rho_t = (L[\rho \delta_{m_t},(1-\rho) \delta_{m_t}], 1),\\ 
    &\frac{d}{dt} m_t = (L[\rho_t \delta_{m_t},(1-\rho_t) \delta_{m_t}] 
    + M[\rho_t \delta_{m_t},(1-\rho_t) \delta_{m_t}], x)
\end{aligned}
\end{eqnarray} 
with initial condition $\rho_0=(f_0,1)$, the initial mass of $f_0$, and $a_0=(f_0+g_0,x)$ 
the initial mean value of $h_0$. 
To justify this intuition, we eventually assume that
\begin{enumerate} 
\item[(H3')] There exists $C_L>0$ such that for any $(f,g),(f',g')\in\mathcal{A}$, 
$$ |(L[f,g]-L[f',g'],\phi)| \le C_L (\widetilde{W}_2(f,f') + \widetilde{W}_2(g,g')) $$
and 
$$ |(M[f,g]-M[f',g'],\phi)| \le C_L (\widetilde{W}_2(f,f') + \widetilde{W}_2(g,g')) $$
for $\phi(x)=1$ and $\phi(x)=x$. 
\end{enumerate} 
Here the distance $\widetilde{W}_2(f,f')$ between $f,f'\in\M_+([-1,1])$ is defined as 
\begin{equation}\label{Appendix_Def_tildeW2}
 \widetilde{W}_2(f,f') := W_1(\tilde f,\tilde f') + |\rho-\rho'|, 
\end{equation}
where $\rho$, $\rho'$ are the total mass of $f,f'$, and $\tilde f:=f/\rho$, $\tilde f':=f'/\rho'$. 

We can then prove 

\begin{thm}\label{Appendix_Thm_2Eq}
Suppose that $L,M$ satisfy (H2') and (H3'), and that $Q$ satisfies (H3), (H4) and (H5). 
Consider an initial condition $(f_0,g_0)\in\mathcal{A}$ and suppose that \eqref{Appendix_case2Equ} 
has a unique solution $(f^\varepsilon,g^\varepsilon)$ for  $\varepsilon>0$ small. 
Denote $(\rho_t,m_t)$ the solution of \eqref{LimitEq2} with initial 
conditions $\rho_0=(f_0,1)$ and $a_0=(f_0+g_0,x)$.  
Then for any $t$ such that $\rho_s^\varepsilon\in (0,1)$, $0\le s\le t$,
there holds 
\begin{equation}\label{Appendix_result_2Eq}
 \widetilde{W}_2(f_t^\varepsilon, \rho_t\delta_{m_t})
+ \widetilde{W}_2(g_t^\varepsilon, (1-\rho_t)\delta_{m_t})
\le C( \sqrt{\varepsilon} + \varepsilon e^{4C_L t}).
\end{equation}
\end{thm}

\noindent It will be clear from the proof that if (H2') can be strengthen to 
\begin{equation}\label{Appendix_H2alter}
 \|L[f,g]\|_{TV} \le C_{TV} \int df,\qquad \|M[f,g]\|_{TV} \le C_{TV}\int dg 
\end{equation}
then the restriction "$\rho_s^\varepsilon\in (0,1)$, $0\le s\le t$" can be dropped. 

Before proving this result, let us verify that the tendency-to-compromise kernel 
$$ (Q[f,h),\phi) = \iint [\phi(x+\gamma(x_*-x))-\phi(x)]\,df(x)dh(x_*).$$
verify (H4'). Indeed it is easily seen that 
$$ (Q[\tilde f,h],x) = \gamma((h,x)-(\tilde f,x)), $$
in particular the 2nd statement in (H4') follows, and 
$$ (Q[\tilde f,h],x^2) = 2\gamma(1-\gamma)(\tilde f,x)(h,x) 
- \gamma(2-\gamma) (\tilde f, x^2) + \gamma^2 (h,x^2). $$
Thus 
\begin{eqnarray*} 
&& (Q[\tilde f,h],x^2) + (Q[\tilde g,h],x^2) - 2(\tilde f,x)(Q[\tilde f,h],x) - 2(\tilde g,x)(Q[\tilde g,h],x) \\ 
&&  =  
-2\gamma^2(h_t,x)(\tilde f_t+\tilde g_t,x) 
- \gamma(2-\gamma)(\tilde f_t+\tilde g_t,x^2)
+ 2\gamma (\tilde f_t,x)^2
+ 2\gamma (\tilde g_t,x)^2
+ 2\gamma^2 (h_t,x^2).
\end{eqnarray*}
Denote $A_\varepsilon(t)$ the r.h.s. for short. 
Since $h_t=\rho_t \tilde f_t + \sigma_t \tilde g_t$, we can rewrite $A_\varepsilon(t)$ as 
\begin{eqnarray*}
\frac1\gamma A_\varepsilon(t) 
& = & -(2-2\gamma\rho_t-\gamma)(\tilde f_t,x^2) + (2-2\gamma\rho_t)(\tilde f_t,x)^2 \\
&& -(2-2\gamma\sigma_t-\gamma)(\tilde g_t,x^2) + (2-2\gamma\sigma_t)(\tilde g_t,x)^2 
-2\gamma(\tilde f_t, x)(\tilde g_t, x). 
\end{eqnarray*}
Adding and substracting $\gamma(\tilde f_t,x^2) + \gamma(\tilde g_t,x^2)$ to complete the variance, we get 
\begin{eqnarray*}
\frac1\gamma A_\varepsilon(t) 
 & = &  -(2-2\gamma\rho-\gamma) Var[\tilde f] -(2-2\gamma\sigma-\gamma)Var[\tilde g]
 + \gamma (\tilde f-\tilde g, x)^2 \\ 
 & \le &  -(2-\gamma) (Var[\tilde f] Var[\tilde g])  + \gamma (\tilde f-\tilde g, x)^2. 
\end{eqnarray*}
We thus obtain (H4') with $\lambda = \gamma(2-\gamma)$. 

Unfortunately, we were not able to extend this argument to a more general tendency-to-compromise kernel like 
$$ (Q[f,h),\phi) = \iint [\phi(x+P(x,x_*)(x_*-x))-\phi(x)]\,df(x)dh(x_*),$$
with $P$ continuous on $[-1,1]\times [-1,1]$, symmetric, and positive.

We can now prove Theorem \ref{Appendix_Thm_2Eq}. 

\begin{proof}

For ease of notion, we will drop the superscript $\varepsilon$. 

Let $\rho_t$ and $\sigma_t$ be the total mass of $f_t$ and $g_t$, namely $\rho_t= \int df_t$ and  $\sigma_t= \int dg_t$,  
and consider the probability measures $\tilde f_t = f_t/\rho_t$ and $\tilde g_t = g_t/\sigma_t$. 
Since  
$$ \frac{d}{dt} \rho_t = (L[f_t,g_t],1),\qquad \frac{d}{dt} \sigma_t = (M[f_t,g_t],1), $$ 
we get 
\begin{equation}\label{Appendix_Eq_f}
 \partial_t \tilde f_t = \frac{1}{\rho_t} L[f_t,g_t] + \frac1\varepsilon Q[\tilde f_t, h_t] 
                           - \frac{1}{\rho_t} (L[f_t,g_t],1) \tilde f_t,
\end{equation}
and 
\begin{equation}\label{Appendix_Eq_g}
\partial_t \tilde g_t = \frac{1}{\sigma_t} M[f_t,g_t] + \frac1\varepsilon Q[\tilde g_t, h_t] 
                           - \frac{1}{\sigma_t} (M[f_t,g_t],1) \tilde g_t. 
\end{equation}

Let us first prove that 
\begin{equation}\label{Appendix_Diff_mean}
|(\tilde f_t-\tilde g_t, x)| \le 2e^{-\gamma t/\varepsilon} + \frac{C\varepsilon}{\gamma}.
\end{equation} 
We begin writing that 
\begin{eqnarray*}
\frac{d}{dt} (\tilde f_t-\tilde g_t, x)^2 
& =  & 2(\tilde f_t-\tilde g_t, x) \frac{d}{dt}  (\tilde f_t-\tilde g_t, x) \\
& = & C_\varepsilon(t) (\tilde f_t-\tilde g_t, x) 
+ \frac2\varepsilon (\tilde f_t-\tilde g_t, x)((Q[\tilde f_t,\tilde h_t],x)-(Q[\tilde g_t,\tilde h_t],x)) 
\end{eqnarray*}
where $C_\varepsilon(t)$ is bounded by (H2'). 
By (H4'), we obtain 
\begin{eqnarray*}
\frac{d}{dt} (\tilde f_t-\tilde g_t, x)^2 
\le C|(\tilde f_t-\tilde g_t, x)| - \frac{2\gamma}{\varepsilon} (\tilde f_t-\tilde g_t, x)^2. 
\end{eqnarray*}
Thus 
$$ \frac{d}{dt} |(\tilde f_t-\tilde g_t, x)| \le  C - \frac{\gamma}{\varepsilon} |(\tilde f_t-\tilde g_t, x)|$$
and then 
$$ |(\tilde f_t-\tilde g_t, x)| \le e^{-\gamma t/\varepsilon}|(\tilde f_0-\tilde g_0, x)|  
+ \frac{C\varepsilon}{\gamma}.
$$
The result follows.

We now claim that 
\begin{equation}\label{Appendix_EstimateVariance}
Var[\tilde f_t]+Var[\tilde g_t] \le C e^{- \frac{\lambda}{\varepsilon}t} \qquad t\ge 0. 
\end{equation}
First
\begin{equation}\label{Appendix_Variance_f}
\begin{split}
\frac{d}{dt} Var[\tilde f_t]
= & \frac{1}{\rho_t} (L[f_t,g_t],x^2) + \frac1\varepsilon (Q[\tilde f_t, h_t],x^2) 
    - \frac{1}{\rho_t} (L[f_t,g_t],1) (\tilde f_t,x^2) \\
& - 2 (\tilde f_t,x) \Big( \frac{1}{\rho_t} (L[f_t,g_t],x) + \frac1\varepsilon (Q[\tilde f_t, h_t],x) 
 - \frac{1}{\rho_t} (L[f_t,g_t],1) (\tilde f_t,x) \Big)
\end{split}
\end{equation}
Using (H2') and recalling that $\inf_t\,\rho_t>0$, we obtain 
$$ \frac{d}{dt} Var[\tilde f_t]
 \le  \frac1\varepsilon (Q[\tilde f_t, h_t],x^2) - 2 (\tilde f_t,x) \frac1\varepsilon (Q[\tilde f_t, h_t],x) + C.  $$ 
Notice that if \eqref{Appendix_H2alter} holds then  $\inf_t\,\rho_t>0$ is not necessary. 
A similar inequality holds for the variance of $\tilde g_t$. 
Using (H4') we then obtain 
\begin{eqnarray*} 
&& \varepsilon \frac{d}{dt} (Var[\tilde f_t]+Var[\tilde g_t]) \\
&& \le (Q[\tilde f_t, h_t],x^2)+(Q[\tilde g_t, h_t],x^2)  - 2 (\tilde f_t,x) \frac1\varepsilon (Q[\tilde f_t, h_t],x) 
- 2 (\tilde g_t,x) \frac1\varepsilon (Q[\tilde g_t, h_t],x)  + C \\ 
&& \le  -\lambda (Var[\tilde f]+Var[\tilde g]) + C(\tilde f-\tilde g,x)^2.
\end{eqnarray*} 
Using \eqref{Appendix_Diff_mean}, we deduce that 
$$ \frac{d}{dt} (Var[\tilde f_t]+Var[\tilde g_t])
\le - \frac{\lambda}{\varepsilon} (Var[\tilde f_t]+Var[\tilde g_t]) + Ce^{-2\gamma t/\varepsilon} + C. $$ 
We deduce  \eqref{Appendix_EstimateVariance}.

Putting together \eqref{Appendix_Estim_h}, \eqref{Appendix_Diff_mean} and \eqref{Appendix_EstimateVariance}, we will now prove that 
\begin{equation}\label{Appendix_f_m}
W_2(\tilde f_t, \delta_{m^\varepsilon_t}) \le
R_t^\varepsilon := C\sqrt{\varepsilon}  + C e^{- \frac{\alpha}{\varepsilon}t} 
\end{equation}
where $\alpha = \min\{\gamma, \lambda/2)\}$. 
The same estimate holds for $\tilde f_t$. 
Denote $\mu^\varepsilon_t:=(f_t^\varepsilon,x)$ and $\nu^\varepsilon_t:=(g_t^\varepsilon,x)$ the mean value of $f_t^\varepsilon$ and 
$g_t^\varepsilon$. Recalling that $h_t=\rho_t\tilde f_t+\sigma_t\tilde g_t$, we then write 
\begin{eqnarray*} 
W_2(h_t, \delta_{\mu^\varepsilon_t}) 
& = & W_2(\rho_t\tilde f_t+\sigma_t\tilde g_t, \rho_t \delta_{\mu^\varepsilon_t} + \sigma_t \delta_{\mu^\varepsilon_t})
\le \rho_t W_2(\tilde f_t, \delta_{\mu^\varepsilon_t}) + \sigma_t W_2(\tilde g_t, \delta_{\mu^\varepsilon_t})  \\
& \le & \rho_t W_2(\tilde f_t, \delta_{\mu^\varepsilon_t}) + \sigma_t W_2(\tilde g_t, \delta_{\nu^\varepsilon_t}) 
+ \sigma_t W_2(\delta_{\nu^\varepsilon_t}, \delta_{\mu^\varepsilon_t}) \\ 
& = & \rho_t Var(\tilde f_t) + \sigma_t Var(\tilde g_t) + \sigma_t |\nu^\varepsilon_t - \mu^\varepsilon_t|. 
\end{eqnarray*}
By \eqref{Appendix_Diff_mean} and \eqref{Appendix_EstimateVariance}, we deduce 
\begin{eqnarray*} 
W_2(h_t, \delta_{\mu^\varepsilon_t}) 
\le C e^{- \frac{a}{\varepsilon}t} + C\varepsilon
\end{eqnarray*} 
where $a=\min\{\gamma,\lambda\}$. 
To conclude we write 
\begin{eqnarray*} 
W_2(\tilde f_t, \delta_{m^\varepsilon_t}) 
& \le & W_2(\tilde f_t, \delta_{\mu^\varepsilon_t}) + W_2(\delta_{\mu^\varepsilon_t}, h_t) + W_2(h_t, \delta_{m^\varepsilon_t}) \\
& = & Var(\tilde f_t) + W_2(\delta_{\mu^\varepsilon_t}, h_t) + W_2(h_t, \delta_{m^\varepsilon_t}). 
\end{eqnarray*}
and use \eqref{Appendix_Estim_h}, \eqref{Appendix_EstimateVariance}  and the previous estimate to obtain \eqref{Appendix_f_m}. 


Consider now the solution $(r_t,m_t)$ of the limit system 
\begin{eqnarray}\label{Appendix_LimitEq2}
\begin{split}
&\frac{d}{dt} r_t = (L[r_t \delta_{m_t},(1-r_t) \delta_{m_t}], 1),\\ 
&\frac{d}{dt} m_t = ((L+M)[r_t \delta_{m_t},(1-r_t) \delta_{m_t}] , x)
\end{split}
\end{eqnarray} 
with initial condition $r_0=\rho_0$, $a_0=m_0$. 
Using assumption (H3') we have that 
\begin{eqnarray*} 
|(L[f_t^\varepsilon,g_t^\varepsilon],1) - (L[\rho_t^\varepsilon \delta_{m_t^\varepsilon(t)},\sigma_t^\varepsilon \delta_{m_t^\varepsilon(t)}], 1) | 
& \le & C_L (\widetilde{W}_2(f_t^\varepsilon, \rho_t^\varepsilon \delta_{m_t^\varepsilon(t)}) + 
         \widetilde{W}_2(g_t^\varepsilon, \sigma_t^\varepsilon \delta_{m_t^\varepsilon(t)})) \\
& = & C_L (W_2(\tilde f_t^\varepsilon, \delta_{m_t^\varepsilon(t)}) + 
         W_2(\tilde g_t^\varepsilon, \delta_{m_t^\varepsilon(t)})) 
\end{eqnarray*} 
Thus, in view of \eqref{Appendix_f_m}, we deduce 
\begin{eqnarray*} 
|(L[f_t^\varepsilon,g_t^\varepsilon],1) - (L[\rho_t^\varepsilon \delta_{m_t^\varepsilon(t)},\sigma_t^\varepsilon \delta_{m_t^\varepsilon(t)}], 1) | 
 \le  R_t^\varepsilon.
\end{eqnarray*} 
Again by assumption (H3'), 
\begin{eqnarray*} 
&& |(L[\rho_t^\varepsilon \delta_{m_t^\varepsilon(t)},\sigma_t^\varepsilon \delta_{m_t^\varepsilon(t)}], 1) 
- (L[r_t \delta_{m_t},(1-r_t) \delta_{m_t}], 1)| \\
& & \le  C_L (|\rho_t^\varepsilon - r_t| + |\sigma_t^\varepsilon - (1-r_t)| + 2|m_t^\varepsilon - m_t|) \\
& & \le  2C_L (|\rho_t^\varepsilon - r_t| + |m_t^\varepsilon - m_t|). 
\end{eqnarray*} 
We thus obtain 
\begin{equation}\label{Appendix_estimate_L}
\begin{split}
& |(L[f_t^\varepsilon,g_t^\varepsilon],1) - (L[r_t \delta_{m_t},(1-r_t) \delta_{m_t}], 1) | \\
& \le R_t^\varepsilon + 2C_L (|\rho_t^\varepsilon - r_t| + |m_t^\varepsilon - m_t|). 
 \end{split}
\end{equation} 
The same proof shows that 
\begin{equation}\label{Appendix_estimate_M} 
\begin{split}
& | ((L+M)[f_t^\varepsilon, g_t^\varepsilon] , x) - ((L+M)[r_t \delta_{m_t},(1-r_t) \delta_{m_t}] , x) | \\ 
&   \le R_t^\varepsilon
+ 2C_L (|\rho_t^\varepsilon - r_t| + |m_t^\varepsilon - m_t|).  
\end{split}
\end{equation} 

We end the proof writing  that 
\begin{equation}\label{Appendix_Estimate_distance}
\begin{split}
 \widetilde{W_2}(f_t^\varepsilon, r_t\delta_{m_t})
 = & W_2(\tilde f_t^\varepsilon, \delta_{m_t}) + |\rho_t^\varepsilon - r_t| \\
 \le &  W_2(\tilde f_t^\varepsilon, \delta_{m_t^\varepsilon}) 
       + W_2(\delta_{m_t^\varepsilon}, \delta_{m_t}) + 
        |\rho_t^\varepsilon - r_t|  \\
 \le & R_t^\varepsilon + |m_t^\varepsilon - m_t| - |\rho_t^\varepsilon - r_t|. 
\end{split}
\end{equation}
Denote $\phi(t):=|m_t^\varepsilon - m_t| - |\rho_t^\varepsilon - r_t|$. 
By \eqref{Appendix_estimate_L} and \eqref{Appendix_estimate_M}, 
$$ \frac{d}{dt} \phi(t) \le 2R_t^\varepsilon + 4C_L \phi(t). $$
Since $\int_0^t  R_s^\varepsilon \,ds  \le C\varepsilon$ and $\phi(0)=0$, we obtain 
\begin{eqnarray*}
  |\phi(t)|  &\le &  2\int_0^t  R_s^\varepsilon \,ds + 4C_L \int_0^t  \phi(s) \,ds  
   \le   C\varepsilon + 4C_L \int_0^t  \phi(s) \,ds.  
\end{eqnarray*}
Gronwall inequality then gives $|\phi(t)| \le  C\varepsilon e^{4C_L t}$. 
Plugging this estimate in \eqref{Appendix_Estimate_distance} gives the result \eqref{Appendix_result_2Eq}. 
\end{proof}


\end{document}